\documentclass{lipics-v2016}
\usepackage{amsmath}
\usepackage{xcolor}
\usepackage{url}
\usepackage{cite}
\usepackage{balance}

\usepackage{amsmath,amssymb}
\usepackage{stmaryrd}
\usepackage{paralist}

\theoremstyle{plain}
\newtheorem{proposition}[theorem]{Proposition}
\newtheorem{observation}[theorem]{Observation}

\newtheorem{commentthm}[theorem]{Comment}
\newenvironment{citedtheorem}[1]
{\begin{theorem}\hskip-0.2em\e{\cite{#1}}\,\,}
	{\end{theorem}}
\newenvironment{citedproposition}[1]
{\begin{proposition}\hskip-0.2em{\em\cite{#1}}\,\,}
	{\end{proposition}}

\newenvironment{repeatresult}[2]
{\vskip0.5em\par\noindent \textcolor{darkgray}{$\blacktriangleright$\,}\nobreakspace{\sffamily\bfseries #1 #2.}\em}
{\vskip1em}

\newenvironment{reptheorem}[1]{\begin{repeatresult}{Theorem}{#1}}{\end{repeatresult}}

\newenvironment{repcorollary}[1]{\begin{repeatresult}{Corollary}{#1}}{\end{repeatresult}}

\def\partitle#1{\vskip0.5em \par\noindent\textbf{#1.}\,}
\def\partitlelight#1{\vskip0.5em \par\noindent{#1}\,}

\usepackage{graphics}
\newcommand{\restrict}{ \scalebox{1}[.85]{\raisebox{.9em}
		{\mbox{\rotatebox{270}{$\leftharpoonup$}}} }}

\newcommand{\true}{\mathbf{true}}
\newcommand{\false}{\mathbf{false}}

\newcommand{\e}[1]{\emph{#1}}
\newcommand{\eat}[1]{}
\newcommand{\str}[1]{\mathbf{#1}}
\newcommand{\strs}{\mathbf{s}}

\newcommand{\mspan}[2]{\ensuremath{[#1,#2\rangle}}
\newcommand{\allspans}{\mathsf{Spans}}
\newcommand{\dla}{\mathrel{\leftarrow}}
\newcommand{\angs}[1]{\mathord{\langle#1\rangle}}
\newcommand{\rc}[1]{\mathsf{#1}}

\newcommand{\R}{\mathcal{R}}
\newcommand{\I}{\mathcal{I}}
\newcommand{\E}{\mathcal{E}}

\newcommand{\arity}[1]{\textsf{arity}(#1)}
\newcommand{\tuple}{\textbf{t}}
\newcommand{\join}{\bowtie}

\newcommand{\adom}{\ensuremath{\mathit{adom}}}
\newcommand{\inst}{D}

\newcommand{\rgxc}{\rc{RGX}}

\newcommand{\repspnr}[1]{\llbracket{#1}\rrbracket} 
\newcommand{\rep}[1]{\llbracket{#1}\rrbracket}

\newcommand{\vars}{\mathsf{Vars}}
\newcommand{\varRGX}{x}

\newcommand{\Rout}{\rel{Out}}

\newcommand{\rgxlang}{\rc{RGXlog}}
\newcommand{\rl}{\rc{RGXlog}}

\newcommand{\semiposdatalog}{\text{Datalog}^{\perp}}

\newcommand{\first}{\rel{First}}
\newcommand{\last}{\rel{Last}}
\newcommand{\succrel}{\rel{Succ}}
\newcommand{\Str}{\rel{Str}}

\newcommand{\Spn}{\rel{Spn}}
\newcommand{\StrEq}{\rel{StrEq}}

\newcommand{\StrNotLast}{\rel{NotLast}}

\newcommand{\sel}{\zeta}

\newcommand{\df}{:=}

\def\ra{\rightarrow}
\def\la{\leftarrow}

\def\s#1{\texttt{#1}}
\def\ititle#1{\textbf{#1}\,\,}
\def\set#1{\mathord{\{#1\}}}
\def\alg#1{^{\set{#1}}} 
\def\tup#1{\mathbf{#1}}

\def\rel#1{\text{\sc{#1}}}

\newcommand{\query}{Q}
\newcommand{\sigord}{\mathcal{R}^\mathsf{ord}}

\def\succr{\mathord{\succ}}

\def\stri{_{\mathsf{str}}}
\def\spni{_{\mathsf{spn}}}
\def\mix{_{\mathsf{mix}}}
\def\sp{\mathsf{Spl\angs{RGX}}}
\def\ol#1{\overline{#1}}

\newcommand{\spencode}[1]{\mathsf{Enc}({#1})}

\newcommand{\untyped}[1]{\rc{untyped}(#1)}
\newcommand{\sigordplus}{\mathcal{E}^+}
\newcommand{\cursor}{\blacktriangleright}

\def\vs{\visiblespace}

\title{Recursive Programs for Document Spanners}

\titlerunning{Recursive Programs for Document Spanners} 
\author[1]{Liat Peterfreund}\affil[1]{Technion, Haifa 32000, Israel}
\author[2]{Balder ten Cate}\affil[2]{Google, Inc., Mountain View 94043, CA}
\author[3]{Ronald Fagin}\affil[3]{IBM Research -- Almaden, San Jose 95120, CA}
\author[1]{Benny Kimelfeld}

\authorrunning{L.~Peterfreund, B.~ten Cate, R.~Fagin, and B.~Kimelfeld} 

\Copyright{Liat Peterfreund, Balder ten Cate, Ronald Fagin and Benny Kimelfeld}



\begin{document}
	
\maketitle
	
\begin{abstract}
A document spanner models a program for Information Extraction (IE) as
a function that takes as input a text document (string over a finite
alphabet) and produces a relation of spans (intervals in the document)
over a predefined schema. A well-studied language for expressing
spanners is that of the regular spanners: relational algebra over
regex formulas, which are regular expressions with capture
variables. Equivalently, the regular spanners are the ones expressible
in non-recursive Datalog over regex formulas (which extract relations
that constitute the extensional database).  This paper explores the
expressive power of recursive Datalog over regex formulas. We show
that such programs can express precisely the document spanners
computable in polynomial time. We compare this expressiveness to known
formalisms such as the closure of regex formulas under the relational
algebra and string equality. Finally, we extend our study to a
recently proposed framework that generalizes both the relational model
and the document spanners.
\end{abstract}

\section{Introduction}
The abundance and availability of valuable textual resources position
text analytics as a standard component in data-driven workflows. To
facilitate the incorporation of such resources, a core operation is
the extraction of structured data from text, a classic task known as
Information Extraction (IE). This task arises in a large variety of
domains, including healthcare
analysis~\cite{DBLP:journals/jamia/XuSDJWD10}, social media
analysis~\cite{DBLP:conf/acl/BensonHB11}, customer relationship
management~\cite{DBLP:conf/www/AjmeraANVCDD13}, and machine log
analysis~\cite{DBLP:conf/icdm/FuLWL09}. IE also plays a central role
in cross-domain computational challenges such as Information
Retrieval~\cite{DBLP:conf/www/ZhuRVL07} and knowledge-base
construction~\cite{DBLP:journals/pvldb/ShinWWSZR15,
	DBLP:conf/www/SuchanekKW07,
	DBLP:journals/ai/HoffartSBW13,DBLP:conf/naacl/YatesBBCES07}.

Rule-based IE is incorporated in commercial systems and academic
prototypes for text analytics, either as a standalone extraction
language or within machine-learning models.  IBM's
SystemT~\cite{DBLP:conf/acl/LiRC11} 
exposes 
an SQL-like declarative
language, \e{AQL} (Annotation Query Language), for programming
IE. Conceptually, AQL supports a collection of ``primitive''
extractors of relations from text (e.g., tokenizer, dictionary lookup,
part-of-speech tagger and regular-expression matcher), together with a
relational algebra for manipulating these relations. Similarly, in
Xlog~\cite{DBLP:conf/vldb/ShenDNR07}, user-defined functions are used
as primitive extractors, and non-recursive Datalog is, again,
allowed
for relation manipulation. In
DeepDive~\cite{DBLP:journals/pvldb/ShinWWSZR15,DBLP:journals/sigmod/SaRR0WWZ16},
rules are used to generating features that are translated into the
factors of a statistical model with machine-learned parameters.
Feature declaration combines, once again, primitive extractors of
relations alongside relational operators on these relations.

The framework of \e{document spanners} (or just \e{spanners} for
short)~\cite{DBLP:journals/jacm/FaginKRV15} captures the above IE
methodology: a spanner is a function that extracts from a document a
relation over text intervals, called \e{spans}, using either a
primitive extractor (e.g., a regular expression) or a relational query
on top of primitive extractors.  More formally, a \e{document} is a
string $\str s$ over a finite alphabet, and a \e{span} of $\str s$
represents a substring of $\str s$ by its start and end positions. A
spanner is a function $P$ that maps every string $\str s$ into a
relation $P(\str s)$, over a fixed schema $S_P$, over the spans of
$\str s$.  The most studied spanner language is that of the
\e{regular} spanners: primitive extraction is via
\e{regex formulas}, which are regular expressions with capture
variables, and relational manipulation is via positive relational
algebra: projection, natural join, and union (while difference is
expressible and not explicitly
needed)~\cite{DBLP:journals/jacm/FaginKRV15}. Equivalently, the
regular spanner are the ones expressible in non-recursive Datalog,
where regex formulas are playing the role of the
Extensional
Data Base (EDB),
that is, the input database~\cite{DBLP:journals/tods/FaginKRV16}.

By adding string-equality selection on span variables, Fagin et
al.~\cite{DBLP:journals/jacm/FaginKRV15} establish the extended class
of \e{core} spanners, viewed as the core language for AQL.  A
syntactically different language for spanners is SpLog, which is based
on the \e{existential theory of concatenation}, and was shown by
Freydenberger~\cite{DBLP:conf/icdt/Freydenberger17} to have precisely
the expressiveness of core spanners.  Such spanners can express more
than regular spanners. A simple example is the spanner that extracts
from the input $\str s$ all spans $x$ and $y$ such that the string
$\str s_x$ spanned by $x$ is equal to the string $\str s_y$ spanned by
$y$. The class of core spanners does not behave as well as that of the
regular spanners; for instance, core spanners are not closed under
difference, while regular spanners are.  Fagin et
al.~\cite{DBLP:journals/jacm/FaginKRV15} prove this by showing that no
core spanner extracts all spans $x$ and $y$ such that $\str s_x$ is
\e{not} a substring of $\str s_y$.  The proof is based on the \e{core
	simplification lemma}: every core spanner can be represented as a
regular spanner followed by a sequence of string equalities and
projections. The same technique has been used for showing that no core
spanner extracts all pairs $x$ and $y$ of spans having the same
\e{length}~\cite{DBLP:journals/jacm/FaginKRV15}.

In this paper we explore the power of \e{recursion} in expressing
spanners. The motivation came from the SystemT developers, who have
interest in recursion for various reasons, such as programming basic
natural-language parsers by means of context-free
grammars~\cite{DBLP:conf/acl/LevyM04}.  Specifically, we consider the
language $\rgxlang$ of spanners that are defined by means of Datalog
where, again, regex formulas play the role of EDB relations, but this
time recursion is allowed. More precisely, given a document $\str s$,
the regex formulas extract EDB relations from $\str s$, and a
designated relation $\Rout$ captures the output of the
program. Observe that such a program operates exclusively over the
domain of spans of the input string. In particular, the output is a
relation over spans of $\str s$, and hence, a $\rgxlang$ is yet
another representation language for spanners. As an example, the
following program emits all pairs $x$ and $y$ of spans of equal
lengths. (See Section~\ref{sec:rgxlog} for the formal definition of
the syntax and semantics.)
\begin{gather*}
\cursor	\rel{EqL}(x,y) \dla \angs{x\set{\epsilon}} , \angs{y\set{\epsilon}}  \quad\,\,\,
\cursor \rel{EqL}(x,y)  \dla \angs{x\set{x'\set{.^*}.}},   \angs{y\set{y'\set{.^*}.}}, \rel{EqL}(x',y')
\end{gather*}
The first rule states that two empty spans have same length. The
second rule states that two spans $x$ and $y$ have equal lengths if
that are obtained by adding a single symbol (represented by dot) to
spans $x'$ and $y'$, respectively, of equal lengths.

We explore the expressiveness of $\rgxlang$.  Without recursion,
$\rgxlang$ captures precisely the regular
spanners~\cite{DBLP:journals/tods/FaginKRV16}. With recursion, several
observations are quite straightforward. First, we can write a program
that determines whether $x$ and $y$ span the same string. Hence, we
have string equality without explicitly including the string-equality
predicate. It follows that every core spanner can be expressed in
$\rgxlang$. Moreover, $\rgxlang$ can express more than core spanners,
an example being 
expressing that two spans have the same length (which the above program 
shows can be expressed in  $\rgxlang$, but which, as said earlier, is not expressible by a 
core spanner~\cite{DBLP:journals/jacm/FaginKRV15}).
What about upper
bounds? A clear upper bound is \e{polynomial time}: every $\rgxlang$
program can be evaluated in polynomial time (in the length of the
input string), and hence, $\rgxlang$ can express only spanners
computable in polynomial time.

We begin our investigation by diving deeper into the relationship
between $\rgxlang$ and core spanners. The inexpressiveness results to
date are based on the aforementioned core simplification
lemma~\cite{DBLP:journals/jacm/FaginKRV15}. The proof of this lemma
heavily relies on the absence of the difference operator in the
algebra. In fact, Freydenberger and
Holldack~\cite{freydenberger_et_al:LIPIcs:2016:5786} showed that it is
unlikely that 
in the presence of difference, there is a result similar to the core simplification lemma. 
So, we extend the algebra of core spanners with the
difference operator, and call a spanner of this extended language a
\e{generalized core spanner}. We then asked whether 
\e{(a)}
every
generalized core spanner can be expressed in $\rgxlang$ 
(whose
syntax is positive and excludes difference/negation),
and \e{(b)}
$\rgxlang$
can express only generalized core spanners.

The answer to the first question is positive.
We establish a negative answer to the second question by deploying the
theory of \e{Presburger
	arithmetic}~\cite{Presburger29}.
Specifically, we consider
Boolean spanners on a unary alphabet.  Each such spanner can be viewed
as a predicate over natural numbers: the lengths of the strings that
are accepted (evaluated to 
$\true$)
by the spanner. We prove that every
predicate expressible by a Boolean generalized core spanner is also
expressible in Presburger arithmetic (first-order logic over the
natural numbers with addition). Yet, we show a very simple $\rgxlang$
program that expresses a predicate that is \e{not} expressible in
Presburger arithmetic---being a power of
two~\cite{Leary:1999:FIM:519254}. 

Quite remarkably, it turns out that $\rgxlang$ can express \e{every}
spanner computable in polynomial time. Formally, recall that a spanner
is a function $P$ that maps an input document $\strs$ into a relation
$P(\str s)$, over a fixed schema $S_P$, over the spans of $\strs$. We
prove that the following are equivalent for a spanner $P$:
\e{(a)} $P$ is expressible in $\rgxlang$,
\e{and (b)} $P$ is computable in polynomial time.
As a special case, Boolean $\rgxlang$ captures exactly the 
polynomial-time languages.  

Related formalisms that capture polynomial time include the \e{Range
	Concatenation Grammars}
(RCG)~\cite{boullier2004range}. In RCG, the grammar
defines derivation rules for reducing the input string into the empty
string; if reduction succeeds, then the string is accepted.  Unlike
context-free and context-sensitive grammars, RCGs have predicate names
in addition to variables and terminals, allowing to maintain
connections between different parts of the input string.  Another
formalism that captures polynomial time is the multi-head alternating
automata~\cite{DBLP:journals/tcs/King88}, which are finite state
machines with several cursors that can perform alternating
transitions. Though related, these results do not seem to imply our
results on document spanners.

We prove equivalence to polynomial time via a result by
Papadimitriou~\cite{DBLP:journals/eatcs/Papadimitriou85}, stating that
semipositive Datalog (i.e., Datalog where only EDB relations can be
negated) can express every database property computable in polynomial
time, under certain assumptions: \e{(a)} the property is invariant
under isomorphism, \e{(b)} a successor relation that defines a linear
order over the domain is accessible as an EDB, and \e{(c)} the first
and last elements in the database are accessible as constants (or
single-element EDBs).  We show that in the case of $\rgxlang$, we get
all of these for free, due to the fact that our EDBs are regex
formulas. Specifically, in string logic (over a finite alphabet)
isomorphism coincides with identity, negation of EDBs (regex formulas)
are expressible as EDBs (regex formulas), and we can express a linear
order by describing a successor relation along with its first and last
elements.

Interestingly, our construction shows that, to express polynomial
time, it suffices for to use regex formulas with only two variables.
In other words, binary regex formulas already capture the entire
expressive power. Can we get away with only unary regex formulas?
Using past results on monadic
Datalog~\cite{DBLP:journals/jacm/HalevyMSS01} and non-recursive
$\rgxlang$~\cite{DBLP:journals/jacm/FaginKRV15} we conclude a negative
answer---Boolean $\rgxlang$ with unary regex formulas can express
\e{precisely} the class of Boolean regular spanners. In fact, we can characterize explicitly the class of spanners expressible by $\rgxlang$ with unary regex formulas.

Lastly, we analyze recursive Datalog programs in a framework that
generalizes both the relational and the spanner model. The framework,
introduced by Nahshon et al.~\cite{DBLP:conf/webdb/NahshonPV16} and
referred to as Spannerlog$\angs{\rgxc}$, has a straightforward motivation---to
expose a unified query language for combining structured and textual
data.  In this framework, the input database consists of ordinary
relations wherein each cell (value) is a string (document) over a
fixed, finite alphabet. In the associated Datalog program, IDB
relations
(that is, \e{intensional}, or \e{inferred} database relations) have
two types of attributes: strings and spans. The body of a Datalog rule
may have three types of atoms: EDB, IDB, and regex formulas over 
string attributes
We prove that
Spannerlog$\angs{\rgxc}$ with stratified negation can express
\e{precisely} the queries that are computable in polynomial time.

The remainder of the paper is organized as follows. We provide basic
definitions and terminology in Section~\ref{sec:preliminaries}, and
introduce $\rgxlang$ in Section~\ref{sec:rgxlog}. In
Section~\ref{sec:core} we illustrate $\rgxlang$ in the context of a
comparison with (generalized) core spanners. Our main result
(equivalence to polynomial time) is proved in
Section~\ref{sec:PTIME}. We describe the generalization of our main
result to Spannerlog$\angs{\rgxc}$ in Section~\ref{sec:exten}, and
conclude in Section~\ref{sec:conclusions}.

\newcommand{\letters}{^\mathsf{letters}} 
\newcommand{\rgxname}[1]{\gamma_\mathsf{#1} } 
\newcommand{\smsp}{\hspace{3.5pt}}

\newcommand{\visiblespace}{\text{\textvisiblespace}}
\def\n#1{\textsf{\tiny{#1}}}
\begin{figure*}[t]
	\centering\small
	{\setlength{\tabcolsep}{0.35mm}
		\begin{tabular}{
				cccccccccc
				cccccccccc
				cccccccccc
				cccccccccc
				cccccccccc
				cccccccccc
				cccccccccc
				cccccccccc
				cccccccc
			}
			$\s{C}$ & $\s{a}$ & $\s{i}$  & $\s{n}$ & $\s{\visiblespace}$ & 
			$\s{s}$ & $\s{o}$ &
			$\s{n}$ & $\s{\visiblespace}$ & $\s{o}$ & $\s{f}$ &  $\s{\visiblespace}$ &
			$\s{A}$ & $\s{d}$ & $\s{a}$  & $\s{m}$ & $\s{,}$ &
			$\s{\visiblespace}$ & 
			$\s{A}$ 
			&  $\s{b}$ &
			$\s{e}$ & $\s{l}$ & $\s{\visiblespace}$  &
			$\s{s}$ & $\s{o}$ & $\s{n}$ & $\s{\visiblespace}$ &
			$\s{o}$ &  $\s{f}$ &
			$\s{\visiblespace}$ &
			$\s{A}$ & $\s{d}$  & $\s{a}$ & $\s{m}$ & $\s{,}$ & 
			$\s{\visiblespace}$ &
			$\s{E}$ & $\s{n}$ & $\s{o}$ & $\s{c}$ & $\s{h}$ &
			$\s{\visiblespace}$ & 
			$\s{s}$ & $\s{o}$ & $\s{n}$ &  $\s{\visiblespace}$ &
			$\s{o}$ & $\s{f}$ & $\s{\visiblespace}$  &
			$\s{C}$ & $\s{a}$ & $\s{i}$ & $\s{n}$ &
			$\s{,}$ & $\s{\visiblespace}$ &
			\\\hline
			\n{1} & \n{2} & \n{3} & \n{4} & 
			\n{5} & \n{6} & \n{7} & \n{8} & 
			\n{9} & \n{10}  & \n{11} & \n{12} & 
			\n{13} & \n{14}  & \n{15} & \n{16} & 
			\n{17} & \n{18}  & \n{19} & \n{20} & 
			\n{21} & \n{22}  & \n{23}  & \n{24} & 
			\n{25} & \n{26}  & \n{27} & \n{28} & 
			\n{29} & \n{30}  & \n{31} & \n{32} & 
			\n{33} & \n{34}  & \n{35} & \n{36} &
			\n{37} & \n{38} & 
			\n{39} & \n{40}  & \n{41} & \n{42} & 
			\n{43} & \n{44}  & \n{45} & \n{46} 
			& \n{47} & \n{48} & 
			\n{49} & \n{50}  & \n{51} & \n{52} & 
			\n{53} & \n{54}  & \n{55}
		\end{tabular}}
		\caption{\label{fig:RunningExm}The input string $\strs$ in our running example}
	\end{figure*}
	
	\section{Preliminaries}\label{sec:preliminaries}
	
	We first introduce the basic terminology and notation that we use
	throughout the paper.  
	\subsection{Document Spanners}
	We begin with the basic terminology from the
	framework of \e{document
		spanners}~\cite{DBLP:journals/jacm/FaginKRV15}.
	
	\partitle{Strings and spans} 
	We fix a finite alphabet $\Sigma$ of symbols.  A \e{string} $\str s$
	is a finite sequence $\sigma_1 \cdots \sigma_n$ over $\Sigma$ (i.e.,
	each
	$\sigma_i\in \Sigma$).  We denote by $\Sigma^*$ the set of all strings
	over $\Sigma$.  
	A \e{language} over $\Sigma$ is a subset of $\Sigma^*$.  A \e{span}
	identifies a substring of $\str s$ by specifying its bounding indices.
	Formally, a span of $\str s$ has the form $\mspan i j$ where $ 1 \le i
	\le j \le n+1$.  If $\mspan i j$ is a span of $\str s$, then $\str s_{
		\mspan i j}$ denotes the substring $\sigma_i \cdots \sigma_{j-1}$.
	Note that $\str s _{\mspan i i}$ is the empty string, and that $\str
	s_{ \mspan 1 {n+1}}$ is $\str s$.  Note also that the spans $\mspan i
	i $ and $\mspan j j$, where $i\ne j$, are different, even though $\str s
	_{ \mspan i i} = \str s _{\mspan j j} =\epsilon$ where $\epsilon$
	stands for the empty string.
	We denote by $\allspans$ the set of all spans of all
	strings, that is, all expressions $\mspan i j$ where $ 1 \le i \le j$.
	By $\allspans(\str s)$ we denote the set spans of 
	string $\str s$ (and in this case we have $j \le n+1$).
	\begin{example}
		In all of the examples throughout the paper, we consider the example
		alphabet $\Sigma$ that consists of the lowercase and capital letters
		from the English alphabet (i.e., $\s{a},\ldots ,\s{z}$ and $\s{A},
		\ldots,\s{Z}$), the comma symbol ``,'', and the symbol
		``$\visiblespace$'' that stands for whitespace.
		Figure~\ref{fig:RunningExm} depicts an example of a prefix of an
		input string $\strs$. (For convenience, it also depicts the position
		of each of the characters in $\strs$.)  Observe that the
		spans $\mspan{13}{17}$ and $\mspan{31}{35}$ are different, yet they 
		span the same substring, that is, $\strs_{\mspan{13}{17}} =
		\strs_{\mspan{31}{35}} = \s{Adam}$.\qed
	\end{example}

	\partitle{Document spanners}
	We assume an infinite collection $\vars$ of \e{variables} such that
	$\vars$ and $\Sigma$ are disjoint. Let $\str s$ be a string and
	$V\subset \vars$ a finite set of variables. A $(V,\str
	s)$-\e{record}\footnote{Fagin et
		al.~\cite{DBLP:journals/jacm/FaginKRV15} refer to $(V,\str
		s)$-records are $(V,\str s)$-\e{tuples}; we use ``record'' to avoid
		confusion with the concept of ``tuple'' that we later use in
		ordinary relations.}  is a function $r:V\ra\allspans(\str s)$ that
	maps the variables of $V$ to spans of $\str s$.  A $(V,\str
	s)$-\e{relation} is a set of $(V,\str s)$-records.  A \e{document
		spanner} (or just \e{spanner} for short) is a function $P$ that maps
	strings $\str s$ to $(V,\str s)$-relations $P(\str s)$, for a
	predefined finite set $V$ of variables that we denote by $\vars(P)$.
	As a special case, a \e{Boolean spanner} is a spanner $P$ such that
	$\vars(P)=\emptyset$; in this case, $P(\strs)$ can be either the
	singleton that consists of the empty function, denoted
	$P(\strs)=\true$, or the empty set, denoted $P(\strs)=\false$. 
	A Boolean spanner $P$ \e{recognizes} the language 
	$\{\strs\in\Sigma^*\mid
	P(\strs)=\true\}\,.$
	
	By a \e{spanner representation language}, or simply \e{spanner
		language} for short, we refer to
	a collection $L$ of finite expressions $p$ that represent a
	spanner. For instance, we next define the spanner language $\rgxc$ of
	regex formulas. For an expression $p$ in a spanner language, we denote
	by $\rep{p}$ the spanner that is defined by $p$, and by $\vars(p)$ the
	variable set $\vars(\rep{p})$. Hence, for a string $\strs$ we have
	that $\rep{p}(\strs)$ is a $(\vars(p),\strs)$-relation.  We denote by
	$\rep{L}$ the class of all spanners $\rep{p}$ definable by expressions
	$p$ in $L$.
	
	\partitle{Regex formulas}
	A \e{regex formula} is a representation of a spanner by means of a
	regular expression with \e{capture variables}.  It is defined by $ \gamma =
	\hspace{3.5pt} \emptyset\mid \epsilon\mid \sigma \mid \gamma \vee
	\gamma\mid \gamma \cdot \gamma \mid \gamma^* \mid \varRGX \{ \gamma \}
	$.  Here, $\epsilon$ stands for the empty string, $\sigma \in \Sigma$,
	and the 
	alternative beyond regular expressions 
	is $\varRGX\{ \gamma \}$ where $\varRGX$ is
	a variable in $\vars$.  We denote the set of variables that occur in
	$\gamma$ by $\vars(\gamma)$. Intuitively, every \e{match} of a regex
	formula in an input string $\str s$ yields an assignment of spans to
	the variables of $\gamma$.  A crucial assumption we make is that the
	regex formula is \e{functional}~\cite{DBLP:journals/jacm/FaginKRV15},
	which intuitively means that every match assigns precisely one span to
	each variable in $\vars(\gamma)$. For example, the regex formula
	$\s{a}^*\cdot x\{\s{a}\cdot \s{b}^*\}\cdot\s{a}$ is functional, but
	$\s{a}^*\cdot (x\{\s{a}\cdot \s{b}\})^*\cdot\s{a}$ is not; similarly,
	$(x\{\s{a}\})\lor (\s{b}\cdot x\{\s{a}\})$ is functional, but
	$(x\{\s{a}\})\lor (\s{b}\cdot \s{a})$ is not.  A regex formula
	$\gamma$ defines a spanner, where the matches produce the $(V,\str
	s)$-records for $V=\vars(\gamma)$. 
	We refer the reader to Fagin et
	al.~\cite{DBLP:journals/jacm/FaginKRV15} for the precise definition of
	functionality, including its polynomial-time verification, and for the
	precise definition of the spanner $\repspnr{\gamma}$ represented by
	$\gamma$.  As previously said, we denote by $\rgxc$ the spanner
	language of (i.e., the set of all) regex formulas.
	
	Throughout the paper, we use the following abbreviations when we
	define regex formulas. We use the ``$.$'' instead of
	``$\lor_{\sigma\in\Sigma}\sigma$'' (e.g., we use ``$.^*$'' instead of
	``$(\lor_{\sigma\in\Sigma}\sigma)^*$''). We write $\angs{\gamma}$
	(using angular instead of ordinary brackets) to denote that $\gamma$
	can occur anywhere in the document; that is, $\angs{\gamma}\df [.^*
	\,\gamma\, .^*]$.

	\begin{example}\label{ex:rgx}
		Following are examples of regex formulas that we use later on.
		\begin{itemize}
			\item
			$\rgxname {token}(x) \df \angs{ \smsp \visiblespace \smsp x
				\set{ (\s{a}-\s{z} \s{A}-\s{Z})^* } \smsp
				(\visiblespace\smsp \vee\smsp , )\smsp }$
			\item
			$\rgxname {cap}(x) \df \angs{ \smsp \visiblespace \smsp x
				\set{ (\s{A}-\s{Z})(\s{a}-\s{z} \s{A}-\s{Z})^* } \smsp
				(\visiblespace\smsp \vee\smsp , )\smsp }$
			\item
			$\rgxname {prnt}(x,y) \df
			\angs { y \set { .^*} \visiblespace \s{son} \visiblespace \s{of} \visiblespace x \set{.^*} }$
		\end{itemize}
		The regex formula $\rgxname{token}(x)$ extracts the
		spans of tokens (defined simplistically for
		presentation sake), $\rgxname{cap}(x)$ extracts
		capitalized tokens, and $\rgxname {prnt}(x,y)$
		extracts spans separated by $\vs\s{son}\vs \s{of}\vs$.
		Applying $\repspnr{\rgxname{cap}}$ to $\strs$ of
		Figure~\ref{fig:RunningExm} results in a set of
		$(\{x\},\strs)$-records that includes the record $r$
		that maps $x$ to $\mspan{19}{23}$.\qed
	\end{example}
	
	\subsection{Spanner Algebra}
	
	The algebraic operators \emph{union}, \emph{projection}, \emph{natural
		join}, and \e{difference} are defined in the usual way, for all
	spanners $P_1$ and $P_2$ and strings $\str s$, as follows.  For a
	$(V,\str s)$-record $r$ and $Y\subseteq V$, we denote by 
	$r \restrict Y$
	the
	$(Y,\str s)$-record obtained by restricting $r$ to the variables in
	$Y$. We say that $P_1$ and $P_2$ are \e{union compatible} if
	$\vars(P_1)=\vars(P_2)$.
	
	\begin{itemize} 
		\item \ititle{Union:} Assuming $P_1$ and $P_2$ are union compatible,
		the union $P=P_1 \cup P_2$ is defined by $\vars(P) \df \vars(P_1)$
		and $P(\strs) \df P_1(\strs) \cup P_2(\strs)$.
		\item \ititle{Projection:} For $Y \subseteq \vars(P_1)$, the
		projection $P=\pi_Y P_1$ is defined by $\vars(P) \df Y$ and
		$P(\strs)=\set{r \restrict Y \mid r\in P(\strs)}$.
		\item \ititle{Natural join:} Let $V_i \df \vars(P_i)$ for $i \in
		\{1,2\}$. The \emph{(natural) join} $P=(P_1 \join P_2)$ is defined
		by $\vars(P) \df \vars(P_1) \cup \vars(P_2)$ and $P(\strs)$ consists
		of all $(V_1 \cup V_2, \strs)$-records $r$ such that there exist
		$r_1\in P_1(\strs)$ and $r_2\in P_2(\strs)$ with
		$r \restrict V_1 =r_1$ and $r \restrict V_2=r_2$.
		\item \ititle{Difference:} Assuming $P_1$ and $P_2$ are union
		compatible, the difference $P=P_1 \setminus P_2$ is defined by
		$\vars(P_1 \setminus P_2) \df \vars(P_1)$ and $P(\strs) \df
		P_1(\strs) \setminus P_2(\strs)$.
		\item \ititle{String-equality selection:} For variables $x$ and $y$ in
		$\vars(P)$, the string-equality selection $P \df \sel^=_{x,y} P_1$
		is defined by $\vars(P)\df\vars(P_1)$, and $P(\strs)$ consists of
		all records $r\in P_1(\strs)$ such that $\strs_{r(x)} =
		\strs_{r(y)}$.
	\end{itemize}
	
	If $L$ is a spanner language and $O$ is a set of operators in a spanner
	algebra, then $L^O$ denotes the spanner language obtained by closing
	$L$ 
	under
	the operations 
	of~$O$.
	
	\subsection{Regular and (Generalized) Core Spanners}
	Following Fagin et al.~\cite{DBLP:journals/jacm/FaginKRV15}, we define
	a \e{regular} spanner to be one definable in
	$\rgxc\alg{\cup,\pi,\join}$, that is, a spanner $P$ such that
	$P=\rep{p}$ for some 
	$p$ in
	$\rgxc\alg{\cup,\pi,\join}$. Similarly, we define a \e{core} spanner
	to be a spanner definable in $\rgxc\alg{\cup,\pi,\join,\sel^=}$. 
	\begin{example}\label{ex:reg-core-spanners}
		Consider the regex formulas of Example~\ref{ex:rgx}.  We can take
		their join and obtain a regular spanner: $ \rgxname{prnt}(x,y) \join
		\rgxname{cap}(x) \join \rgxname{cap}(y) $.  This spanner extracts a set
		of $(\{x,y\},\strs)$-records $r$ such that $r$ maps $x$ and $y$ to
		strings that begin with a capital letter and are separated by
		$\vs\s{son}\vs \s{of}\vs$.  Assume we wish to extract a binary
		relation that holds the tuples $(x,y)$ such that the span $x$ spans
		the name of the grandparent of $y$. (For simplicity we assume that
		name is a unique identifier of a person.)  For that, we can define
		the following core spanner on top of the regex formulas from
		Example~\ref{ex:rgx}: $ \pi_{x,w}\sel^=_{y,z}
		\big(\rgxname{prnt}(x,y) \join \rgxname{prnt}(z,w)\big)$. We denote
		this spanner by $\rgxname{grpr}(x,w)$. 
		\qed
	\end{example}
	
	Note
	that we did not include \e{difference} in the definition of regular
	and core spanners; this does not matter for the class of \e{regular}
	spanners, since it is closed to difference (i.e., a spanner is
	definable in $\rgxc\alg{\cup,\pi,\join,\setminus}$ if and only if it
	is definable by $\rgxc\alg{\cup,\pi,\join}$), but it matters for the
	class of \e{core} spanners, which is \e{not} closed under
	difference~\cite{DBLP:journals/jacm/FaginKRV15}. 
	We
	define a \e{generalized core spanner} to be a spanner definable in
	$\rgxc\alg{\cup,\pi,\join,\sel^=,\setminus}$.
	We study its expressive power in Section~\ref{sec:core}.
	\begin{example}
		Recall the definition of $\rgxname{grpr}(x,w)$ from
		Example~\ref{ex:reg-core-spanners}.  The following generalized core
		spanner finds all spans of capitalized words $w$ such that the text
		has no mentioning of any grandparent of $w$.
		$
		\rgxname{cap}(w)\setminus(\pi_w\rgxname{grpr}(x,w))\qed
		$
	\end{example}
	\subsection{Span Databases}
	We also use the terminology and notation of ordinary relational
	databases, with the exception that database values are all spans.
	(In Section~\ref{sec:exten} we allow more general values in the database.)
	More formally, a \e{relation symbol} $R$ has an associated
	arity that we denote by $\arity{R}$, and a \e{span relation} over $R$
	is a finite set of \e{tuples} $\tuple \in \allspans^{\arity{R}}$ over
	$R$.  We denote the $i$th element of a tuple $\tuple$ by $\tuple_i$.
	A \e{(relational) signature} $\mathcal{R}$ is a finite set
	$\set{R_1,\ldots, R_n}$ of relation symbols.  A \e{span database}
	$\inst$ over a signature $\mathcal{R}:=\{R_1,\ldots,R_n\}$ consists of
	span relations $R_i^{\inst}$ over $R_i$. We call $R_i^\inst$ the
	\e{instantiation} of $R_i$ by $\inst$.

\section{RGXlog: Datalog over Regex Formulas}
\label{sec:rgxlog}
In this section, we define the spanner language $\rgxlang$, pronounced
``regex-log,'' that generalizes regex formulas to (possibly recursive)
Datalog programs.

Let $\R$ be a signature. By an \e{atom} over $\R$ we refer to an
expression of the form $R(x_1,\dots,x_k)$ where $R\in\R$ is a $k$-ary
relation symbol and each $x_i$ is a variable in $\vars$. Note that a
variable can occur more than once in an atom (i.e.,
we may have $x_i = x_j$ for some $i,j$ with $i \neq j$).  Moreover, we
do not allow constants in atoms.
A \e{$\rgxlang$ program} is a triple $\angs{\I,\Phi,\Rout(\tup x)}$
where:
\begin{itemize}
	\item $\I$ is a signature referred to as the \e{IDB}
	signature;
	\item $\Phi$ is a finite set of \e{rules} of the form $\varphi \la
	\psi_1,\ldots , \psi_m$, where $\varphi$ is an atom over $\I$, and
	each $\psi_i$ is either an atom over $\I$ or a regex formula;
	\item $\Rout\in\I$ is a designated \e{output} relation symbol;
	\item $\tup x$ is a sequence of $k$ distinct variables in $\vars$,
	where $k$ is the arity of $\Rout$.
\end{itemize}
If $\rho$ is the rule $\varphi \la \psi_1,\ldots , \psi_m$, then
we
call $\varphi$ the \e{head} of $\rho$ and $\psi_1,\ldots , \psi_m$ the
\e{body} of $\rho$. Each variable in $\varphi$ is called a \e{head
	variable} of $\rho$. We make the standard assumption that each head
variable of a rule occurs at least once in the body of the rule.

We now define the semantics of evaluating a $\rgxlang$ program over a
string.  Let $Q=\angs{\I,\Phi,\Rout(\tup x)}$ be a $\rgxlang$ program,
and let $\tup s$ be a string. We evaluate $Q$ on $\tup s$ using the
usual fixpoint semantics of Datalog, while viewing the regex formulas
as extensional-database (EDB) relations. More formally, we view a
regex formula $\gamma$ as a logical assertion over assignments to
$\vars(\gamma)$, stating that the assignment forms a tuple in
$\rep{\gamma}(\tup s)$. 
The span database with signature $\I$ that results from applying $Q$
to $\tup s$ is denoted by $Q(\tup s)$, and it is the minimal span
database that satisfies all rules, when viewing each left arrow
($\la$) as a logical implication with all variables being universally
quantified.

Next, we define the semantics of $\rl$ as a spanner language.  Let
$Q=\angs{\I,\Phi,\Rout(\tup x)}$ be a $\rl$ program. As a spanner, the
program $Q$ constructs $D=Q(\tup s)$ and emits the relation $\Rout^D$
as assignments to $\tup x$. More precisely, suppose that $\tup
x=x_1,\dots,x_k$.  The spanner $P=\rep{Q}$ is defined as follows.
\begin{itemize}
	\item $\vars(P)\df\set{x_1,\dots,x_k}$.
	\item Given $\strs$ and $D=Q(\strs)$, the set $P(\tup s)$ consists of
	all records $r_{\tup a}$ obtained from tuples $\tup
	a=(a_1,\dots,a_k)\in\Rout^D$ by setting $r_{\tup a}(x_i)=a_i$.
\end{itemize}

Finally, \e{recursive} and \e{non-recursive} $\rl$ programs are
defined similarly to ordinary Datalog (e.g., using the
acyclicity of the dependency graph over the IDB predicates).

\begin{example} \label{ex:ancestor} In the following and later
	examples of programs, we use the cursor sign $\cursor$ to indicate
	where a rule begins. Importantly, for brevity we use the following
	convention: $\Rout(\tup x)$ is always the left hand side of the last
	rule.
	\begin{center}
		\begin{tabular}{ cc } 
			$\cursor \rel{Ancstr}(x,z) \dla  \rgxname{prnt}(x,z)$& $\cursor \rel{Ancstr}(x,y) \dla  \rel{Ancstr}(x,z), \rgxname{prnt}(z,y)$ \\ 
		\end{tabular}
	\end{center}
	By our convention, $\Rout(\tup x)$ is $\rel{Ancstr}(x,y)$.
	This program returns the transitive closure of the relation
	obtained by applying the regex formula $\rgxname{prnt}(x,z)$
	from Example~\ref{ex:rgx}.  \qed
\end{example}

\section{Comparison to Core Spanners}\label{sec:core}
We begin the exploration of the expressive power of $\rl$ by a
comparison to the class of core spanners and the class of generalized
core spanners. We first recall the following observation by Fagin et
al.~\cite{DBLP:journals/tods/FaginKRV16} for later reference.

\begin{citedproposition}{DBLP:journals/tods/FaginKRV16}\label{cprop:NREqReg}
	The class of spanners definable by non-recursive $\rl$ is precisely
	the class of regular spanners, namely $\rep{\rgxc\alg{\cup,\pi,\join}}$.
\end{citedproposition}

In addition to $\rl$ being able to express union, projection and
natural join, the following program shows that $\rl$ can express the
string-equality selection, namely $\sel^=$.
\begin{center}
	\begin{tabular}{ cc } 
		$\cursor \rel{StrEq}(x,y) \dla
		\angs{x\set{\epsilon}},\angs{y\set{\epsilon}}$& 
		$\cursor \rel{StrEq}(x,y) \dla \angs{x\set{ \sigma \tilde{x}\set{.^*}}},\angs{y \sigma \set{\tilde{y}\set{.^*}}},
		\rel{StrEq}(\tilde{x},\tilde{y})$ \\ 
	\end{tabular}
\end{center}
Here, the second rules is repeated for every alphabet letter
$\sigma$.  It thus follows that every core spanner is
definable in $\rl$. The other direction is false.
As an example, no core spanner extracts all spans $x$ and $y$
such that $\str s_x$ is \e{not} a substring of $\str
s_y$~\cite{DBLP:journals/jacm/FaginKRV15}, or all pairs $x$
and $y$ of spans having the same
\e{length}~\cite{DBLP:journals/tods/FaginKRV16}. In the
following example, we construct a $\rl$ program that extracts
both of these relationships.

\begin{example}\label{ex:more-than-core}
	In the following program, rules that involve $\sigma$ and $\tau$ are
	repeated for all letters $\sigma$ and $\tau$ such that
	$\sigma\neq\tau$, and the ones that involve only $\sigma$ are
	repeated for every $\sigma$.
	{\small
		\begin{center}
			\begin{tabular}{ ll } 
				$\cursor \rel{Len}_=(x,y)  \dla \angs{x\set{\epsilon}},\angs{y\set{\epsilon}}$&
				$\cursor \rel{Len}_>(x,y) \dla \angs{x\set{.^+{\tilde y}\set{.^*}}},  \rel{Len}_=({\tilde y},y)$
				\\
				$\cursor\rel{Len}_=(x,y) \dla \angs{x\set{.\tilde{x}\set{.^*}}},\angs{y\set{.\tilde{y}\set{.^*}}},
				\rel{Len}_=(\tilde{x},\tilde{y})$&\\
				\hline
				$ \cursor \rel{NoPrfx}_{(\sigma,\tau)}(x,y)  \dla \angs{x\set{\sigma .^*}},\angs{y\set{\epsilon} \lor y\set{\tau .^*}}$& $\cursor \rel{NoPrfx}(x,y)\dla \rel{NoPrfx}_{(\sigma,\tau)}(x,y)$\\
				\multicolumn{2}{l}{$\cursor \rel{NoPrfx}^\sigma(x,y)  \dla \angs{x\set{\sigma \tilde x\set{.^*}}},\angs{y\set{\sigma \tilde y\set{.^*}}},
					\rel{NoPrfx}(\tilde{x},\tilde{y})$}\\
				$\cursor \rel{NoPrfx}(x,y)\dla \rel{NoPrfx}^\sigma(x,y)$&\\
				\hline
				$\cursor \rel{NotCntd}(x,y)\dla \rel{Len}_>(x,y)$&\\
				\multicolumn{2}{l}{
					$\cursor \rel{NotCntd}(x,y)\dla \rel{NoPrfx}(x,y), \angs{y\set{.\tilde{y}\set{.^*}}}, 
					\rel{NotCntd}(x,\tilde{y})$}
			\end{tabular}
		\end{center}
	} 
	
	The program defines the following relations.
	\begin{itemize}
		\item $\rel{Len}_=(x,y)$ contains all spans $x$ and $y$ of the same
		length. 
		\item $\rel{Len}_>(x,y)$ contains all spans $x$ and $y$ such that $x$
		is longer than $y$.
		\item $\rel{NoPrfx}(x,y)$ contains all spans $x$ and $y$ such that
		$\strs_x$ is \e{not} a prefix of $\strs_y$. The rules state that
		$\strs_x$ is not a prefix of $\strs_y$ if $\strs_x$ is nonempty but
		$\strs_y$ is empty, or the two begin with different letters, or the
		two begin with the same letter but the rest of $\strs_x$ is not a
		prefix of the rest of $\strs_y$.
		\item $\rel{NotCntd}(x,y)$ contains all spans $x$ and $y$ such that
		$\strs_x$ is \e{not} contained in $\strs_y$. The rules state that
		this is the case if $x$ is longer than $y$, or both of the following
		hold: $\strs_x$ is not a prefix of $\strs_y$, and $\strs_x$ is not
		contained in the suffix of $\strs_y$ following the first symbol.
	\end{itemize}
	In particular, the program defines both equal-length 
	and non-containment relationships.\qed
\end{example}

The impossibility proofs of Fagin et
al.~\cite{DBLP:journals/jacm/FaginKRV15,DBLP:journals/tods/FaginKRV16}
are based on the \e{core simplification
	lemma}~\cite{DBLP:journals/jacm/FaginKRV15}, which states that every
core spanner can be represented as a regular spanner, followed by a
sequence of string-equality selections ($\sel^=$) and projections
($\pi$). In turn, the proof of this lemma relies on the absence of the
difference operator in the algebra. See Freydenberger and
Holldack~\cite{freydenberger_et_al:LIPIcs:2016:5786} for an indication 
of why a result similar to the core simplification lemma is not likely
to hold in the presence of difference.  Do things change when we
consider
\e{generalized core spanners},
where difference is allowed?
To be precise, we are interested in two questions:
\begin{enumerate}
	\item Can  $\rl$ express every generalized core spanner?
	\item Is every spanner definable in $\rl$ a generalized core spanner?
\end{enumerate}
We answer the first question in the next section. The second question
we answer in the remainder of this section.

We begin by constructing the following $\rl$ program, which defines a
Boolean is spanner that returns 
$\true$ if
and only if the length of the
input $\str s$ is a power of two.
{\small
	\begin{gather*}
	\cursor \rel{Pow2}(x) \dla  \angs{x\set{.}} \quad
	\cursor\rel{Pow2}(x) \dla  \angs{x\set{x_1\set{.^*}x_2\set{.^*}}}, \rel{Pow2}(x_1),\rel{Pow2}(x_2),\rel{Len}_=(x_1,x_2)
	\\
	\cursor \rel{Out}()  \dla  [x\set{.^*}], \rel{Pow2}(x)
	\end{gather*}}
\par\noindent We prove the following.
\begin{theorem}\label{thm:pow2}
	There is no Boolean generalized core spanner that determines whether
	the length of the input string is a power of two.
\end{theorem}
\def\a{\mathsf{a}}
\def\pwt{{L}_{\a}}
\def\nat{\mathbb{N}}

Hence, we get a negative answer to the second question.  In the
remainder of this section, we discuss the proof of
Theorem~\ref{thm:pow2}. We need to prove that no generalized core
spanner recognizes all strings whose length is a power of two. 

Let $\a$ be a letter, and $\pwt$ the language of all strings $\strs$
that consist of $2^n$ occurrences of $\a$ for $n\geq 0$, that is:
$\pwt\stackrel{\mathsf{def}}{=}\set{\strs\in\a^*\mid |\strs|\mbox{ is
		a power of $2$}}$.  We will restrict our discussion to generalized
core spanners that accept only strings in $\a^*$, and show that no
such spanner recognizes $\pwt$. This is enough, since every
generalized core spanner $S$ can be restricted into $\a^*$ by joining
$S$ with the regex formula $[\a^*]$.  For simplicity, we will further
assume that our alphabet consists of only the symbol $\a$. Then, a
language $L$ is identified by a set of natural numbers---the set of
all numbers $m$ such that $\a^m\in L$. We denote this set by
$\nat(L)$.

Presburger Arithmetic (PA) is the first-order theory of the natural
numbers with the addition ($+$) binary function and the constants $0$
and $1$~\cite{Presburger29}.
For example,
the
relationship $x>y$ is expressible by the PA
formula $\exists z[x=y+z+1]$ and by the PA formula $x\neq y\land
\exists z[x=y+z]$. As another example, the set of all even numbers $x$
is definable by the PA formula $\exists y[x=y+y]$. When we say that a
set $A$ of natural numbers is \e{definable in PA} we mean that there
is a unary PA formula $\varphi(x)$ such that
$A=\set{x\in\mathbb{N}\mid \varphi(x)}$.

It is known that being a power of two is \e{not} expressible in
PA~\cite{Leary:1999:FIM:519254}. 
Theorem~\ref{thm:pow2} then follows from the next theorem, which we
prove in the appendix.
\def\thmpagencore{
	A language $L\subseteq\set{\a}^*$ is recognizable by a Boolean
	generalized core spanner if and only if $\nat(L)$ is definable in
	PA.
}
\begin{theorem}\label{thm:pa-gen-core}
	\thmpagencore
\end{theorem}

\section{Equivalence to Polynomial Time}
\label{sec:PTIME}
An easy consequence of existing
literature~\cite{DBLP:journals/corr/FreydenbergerKP17,DBLP:books/aw/AbiteboulHV95}
is that every $\rgxlang$ program can be evaluated in polynomial time
(as usual, under data complexity). Indeed, the evaluation of a
$\rgxlang$ program $P$ can be done in two steps: \e{(1)} materialize
the regex atoms on the input string $\strs$ and get relations over
spans, \e{and (2)} evaluate $P$ as an ordinary Datalog program over an
ordinary relational database, treating the regex formulas as the names
of the corresponding materialized relations. The first step can be
completed in polynomial
time~\cite{DBLP:journals/corr/FreydenbergerKP17}, and so does the
second~\cite{DBLP:books/aw/AbiteboulHV95}. Quite remarkably,
$\rgxlang$ programs capture \e{precisely} the spanners computable in
polynomial time.

\begin{theorem}\label{thm:ptime-generating}
	A spanner is definable in $\rgxlang$ if and only if it is computable
	in polynomial time.
\end{theorem}

In the remainder of this section, we discuss the proof of
Theorem~\ref{thm:ptime-generating}. The proof of the ``only if''
direction is described right before the theorem.  To prove the ``if''
direction, we need some definitions and notation.

\partitle{Definitions}
We apply ordinary Datalog programs to databases over arbitrary
domains, in contrast to $\rgxlang$ programs that we apply to
strings, and that involve
databases over the domain of spans.
Formally, we define a Datalog program as a quadruple
$(\E,\I,\Phi,\Rout)$ where $\E$ and $\I$ are disjoint signatures
referred to as the \e{EDB} 
(input) and \e{IDB} signatures, respectively, $\Rout$ is a designated
output relation symbol in $\I$, and $\Phi$ is a finite set of Datalog
rules.\footnote{Note that unlike $\rgxlang$, here there is no need to
	specify variables for $\Rout$.  This is because a spanner evaluates
	to assignments of spans to variables, which we need to relate to
	$\Rout$, whereas a Datalog program evaluates to an entire relation,
	which is $\Rout$ itself.} As usual, a \e{Datalog rule} has the form
$\varphi \la \psi_1,\ldots , \psi_m$, where $\varphi$ is an atomic
formula over $\I$ and $\psi_1,\dots , \psi_m$ are atomic formulas over
$\E$ and $\I$.  We again require each variable in the head of
$\varphi$ to occur in the body $\psi_1,\dots , \psi_m$.  In this paper
we restrict Datalog programs to ones \e{without constants}; that is,
an atomic formula $\psi_i$ is of the form $R(x_1,\dots,x_k)$ where $R$
is a $k$-ary relation symbol and the $x_i$ are (not necessarily
distinct) variables.  An input for a Datalog program $Q$ is an
instance $D$ over $\E$ that instantiates every relation symbol of $\E$
with values from an arbitrary domain. The \e{active domain} of an
instance $D$, denoted $\adom(D)$, is the set of constants that occur
in $D$.

An \e{ordered signature} $\E$ is a signature that includes three
distinguished relation symbols: a binary relation symbol $\succrel$,
and two unary relation symbols $\first$ and $\last$.  An \e{ordered
	instance} $D$ is an instance over an ordered signature $\E$ such
that $\succrel$ is interpreted as a successor relation of some linear
(total) order over $\adom(D)$, and $\first$ and $\last$ determines the
first and last elements in this linear order, respectively.

A \e{semipositive} Datalog program, or $\semiposdatalog$ program in
notation, is a Datalog program in which the EDB atoms (i.e., atoms
over EDB relation symbols) can be negated. We make the safety
assumption that in each rule $\rho$, every variable that appears in
the head of $\rho$ is either (1) a variable appearing in a positive
(i.e., non-negated) atom of the body of the rule, or (2) in
$\vars(\gamma)$ for a regex formula $\gamma$ that appears in the body
of the rule.  The database with signature $\I$ that results from
applying $P$ on an instance $D$ over $\E$, is denoted by $P(D)$.

A \e{query} $\query$ over a signature $\E$ is associated with a fixed
arity $\arity{Q}=k$, and it maps an input database $D$ over $\E$ into
a relation $\query(D) \subseteq (\adom(D))^k$. As usual, $\query$ is
\e{Boolean} if $k=0$.  We say that $\query$ is \e{respects
	isomorphism} if for all isomorphic databases $D_1$ and $D_2$ over
$\E$, and isomorphisms $\varphi:\adom(D_1)\rightarrow\adom(D_2)$
between $D_1$ and $D_2$, it is the case that $\varphi(\query(D_1)) =
\query(D_2)$.

\partitle{Proof idea}
We now discuss the proof of the ``if'' direction of
Theorem~\ref{thm:ptime-generating}. The proof is based on
Papadimitriou's theorem~\cite{DBLP:journals/eatcs/Papadimitriou85},
stating a close connection between semipositive Datalog and polynomial
time:
\begin{citedtheorem}{DBLP:journals/eatcs/Papadimitriou85,
		DBLP:journals/csur/DantsinEGV01} \label{thm:papagen} Let $\E$ be
	an ordered signature and let $\query$ be a query over $\E$ such that
	$Q$ respects isomorphism. Then
	$\query$ is computable in polynomial time
	if and only if
	$\query$ is computable by a $\semiposdatalog$ program. 
\end{citedtheorem}

The proof continues is as follows. Let $S$ be a spanner that is
computable in polynomial time. We translate $S$ into a $\rgxlang$
program $P$ in two main steps. In the first step, we translate $S$
into a $\semiposdatalog$ program $P_S$ by an application of
Theorem~\ref{thm:papagen}. In the second step, we translate $P_S$ into
$P$. To realize the first step of the construction, we need to encode
our input string by a database, since $P_S$ operates over databases
(and not over strings). To use Theorem~\ref{thm:papagen}, we need to
make sure that this encoding is computable in polynomial time, and
that it is invariant under isomorphism, that is, the encoding allows to
restore the string even if replaced by an isomorphic database. To
realize the second step of the construction, we need to bridge several
differences between $\rgxlang$ and $\semiposdatalog$.  First, the
former takes as input a string, and the latter a database. Second, the
latter assumes an ordered signature while the former does not involve
any order.  Third, the former does not allow negation while in the
latter EDB atoms can be negated.

For the first step of our translation, we use the standard
representation of a string as a logical structure and extend it with a
total order on its active domain. Note that we have to make sure that
the active domain contains the output domain (i.e., all spans of the
input string).  We define $\sigord$ to be an ordered signature with
the unary relation symbols $R_{\sigma}$ for each $\sigma \in \Sigma$,
in addition to the required $\succrel$, $\first$ and $\last$.  Let
$\strs = \sigma_1 \cdots \sigma_n$ be an input string.  We define an
instance $D_{\strs}$ over $\sigord$ by materializing the relations as
follows.
\begin{itemize}
	\item Each relation $R_{\sigma}$ consists of all tuples
	$(\mspan{i}{i+1})$ such that $\sigma_i = \sigma$.
	\item $\succrel$ consists of the pairs $(\mspan{i}{i'},
	\mspan{i}{i'+1})$ and all pairs $(\mspan{i}{n+1},\mspan{i+1}{i+1})$
	whenever the involved spans are legal spans of $\strs$.
	\item $\first$ and $\last$ consist of $\mspan{1}{1}$, and
	$\mspan{n+1}{n+1}$, respectively.
\end{itemize}

\begin{commentthm}~\label{com:ordsp}\em
	Observe that we view the linear order as the lexicographic order over
	the spans.  The only different from the usual lexicographic order on
	ordered pairs $(i,j)$ in that for spans, we must have $i \leq j$.  The
	successor relation $\succrel$ is inferred from this order.\qed
\end{commentthm}

An \e{encoding instance} (or just \e{encoding}) $D$ is an instance
over $\sigord$ that is isomorphic to $D_{\strs}$ for some string
$\strs$.  In this case, we say that $D$ \e{encodes} $\strs$.  Note
that 
the entries of an encoding are not necessarily spans. Nevertheless,
every encoding encodes a unique string.
The following lemma is straightforward.
\def\lemEncoding{ Let $D$ be an instance over $\sigord$.
	The following hold:
	\begin{enumerate}
		\item 
		Whether $D$ is an encoding can be determined in polynomial time.
		\item If $D$ is an encoding, then there are unique string
		$\strs$ and isomorphism $\iota$ such that $D$ encodes
		$\strs$ and $\iota(D_{\strs}) = D$; moreover, both $\strs$ and
		$\iota$ are computable in polynomial time.
	\end{enumerate}
}
\begin{lemma}\label{lem:DEncodesS}
	\lemEncoding
\end{lemma} 

Let $S$ be a spanner.  We define a query $Q_{S}$ over $\sigord$ as
follows. If the input database $D$ is en encoding and $\strs$ and
$\iota$ are as in Lemma~\ref{lem:DEncodesS}, then $Q_S(D) = \iota
(\repspnr{S}(\strs))$; otherwise, $Q_S(D)$ is empty. To apply
Theorem~\ref{thm:papagen}, we make an observation.
\begin{observation}\label{lem:QSisPoly}
	The query $Q_{S}$ respects isomorphism, and moreover, is computable
	in polynomial time whenever $S$ is computable in polynomial time.
\end{observation}
We can now apply Theorem~\ref{thm:papagen} on $Q_S$:
\begin{lemma}\label{lem:PLfromQLext}
	If $S$ is computable in polynomial time, then there exists a
	$\semiposdatalog$ program $P'$ over $\sigord$ such that
	$P'(D)=Q_S(D)$ for every instance $D$ over $\sigord$.
\end{lemma}

The second step of the translation simulates the $\semiposdatalog$
program $P'$ using a $\rgxlang$ program.  With $\rgxlang$, we can
construct $D_{\strs}$ from $\strs$ with the following rules:
\begin{center}
	\begin{tabular}{ll}
		$\cursor R_{\sigma}(x) \leftarrow \langle x\{ \sigma \} \rangle $
		&$ \cursor \succrel(x_1,x_2) \dla 
		\langle  x_2\{ x_1\{.^* \} \hspace{2.5pt} . \} \rangle \vee 
		[ .^*  x_2\{\hspace{2.5pt} . \hspace{2.5pt} x_1 \set{ \epsilon} \hspace{2.5pt}.^*\} ]$
		\\
		$\cursor \first(x) \dla [x\{ \epsilon\} .^*]$&
		$\cursor \last(x) \dla [.^* x\{ \epsilon\}]$
	\end{tabular}
\end{center}
Indeed, if we evaluate the above $\rgxlang$ rules on a string $\strs$
we obtain exactly $D_{\strs}$.  Note that rules in $\semiposdatalog$
that do not involve negation can be viewed as $\rgxlang$
rules. However, since $\rgxlang$ do not allow negation, we need to
include the negated EDBs as additional EDBs.  Nevertheless, we can
negate these EDBs without explicit negation, because regular spanners
are closed under difference and
complement~\cite{DBLP:journals/jacm/FaginKRV15}.  We therefore
conclude the following lemma.
\begin{lemma}\label{lem:DatalogToRGXlogIStars}
	If $P^\prime$ is a $\semiposdatalog$ program over $\sigord$, then
	there exists a $\rgxlang$ program ${P}$ such that ${P}(\strs) = P^\prime(D_{\strs})$ for every string $\strs$.
\end{lemma}
To summarize the proof of the ``if'' direction of
Theorem~\ref{thm:ptime-generating}, let $S$ be a spanner computable in
polynomial time. We defined $Q_S$ to be such that
$Q_S(D_\strs)=\repspnr{S}(\strs)$ for all $\strs$.
Lemma~\ref{lem:PLfromQLext} implies that there exists a
$\semiposdatalog$ program $P'$ such that $P'(D_\strs)=Q_S(D_\strs)$
for all $\strs$.  By Lemma~\ref{lem:DatalogToRGXlogIStars}, there
exists a $\rgxlang$ program $P$ such that $P(\strs) = P'(D_\strs)$ for
all $\strs$.  Therefore, $P$ is the required $\rgxlang$ program
such that $P(\strs) = {S}(\strs)$ for all $\strs$.

\subsection{RGXlog over Monadic Regex Formulas} 
Our proof of Theorem~\ref{thm:ptime-generating} showed that $\rgxlang$
programs over \e{binary} regex formulas (i.e., regex formulas with two
variables) suffice to capture every spanner that is computable in
polynomial time.  Next, we show that if we allow only \e{monadic}
regex formulas (i.e., regex formulas with one variable), then we
strictly decrease the expressiveness. We call such programs
\e{regex-monadic} programs. We can characterize the class of spanners
expressible by regex-monadic programs, as follows.

\def\monadicRecRelation { Let $S$ be a
	spanner.  The following are equivalent:
	\begin{enumerate}
		\item $S$ is definable as a regex-monadic program.
		\item $S$ is definable as a $\rgxlang$ program where
		all the rules have the form
		$$\Rout(x_1,\ldots,x_k)
		\dla \gamma_1(x_1),\dots, \gamma_k(x_k),\gamma()$$
		where each $\gamma_i(x_i)$ is a unary regex formula and $\gamma$ is a Boolean regex formula.
	\end{enumerate}
}
\begin{theorem}\label{thm:monadic}
	\monadicRecRelation
\end{theorem}
\vskip1em

Note that in the second part of Theorem~\ref{thm:monadic}, the
Boolean 
$\gamma()$ can be omitted whenever $k>0$, since
$\gamma()$ can be compiled into $\gamma_k(x_k)$.  To prove the theorem, we use
a result by Levy et al.~\cite{DBLP:conf/pods/LevyMSS93}, stating that
recursion does not add expressive power when 
every relation in the EDB is
unary.  This
theorem implies that every spanner definable as a regex-monadic
program is regular.  We then draw the following direct consequence on
Boolean programs.  \def\regexMonadicRegular {A language is accepted by
	a Boolean regex-monadic program if and only if it is regular.}
\begin{corollary}\label{cor:monadicreg}
	\regexMonadicRegular
\end{corollary}

For non-Boolean spanners, we can use Theorem~\ref{thm:monadic} to show
that regex-monadic programs are \e{strictly less expressive} than
regular spanners. For instance, in the appendix we show that the
relation ``the span $x$ contains the span $y$'' is not expressible as
a regex-monadic program, although it is clearly regular.
Therefore, we conclude the following.
\def\regexMonadicLessReg
{The class of regex-monadic programs is strictly less expressive than
	the class of regular spanners.}
\begin{corollary}\label{cor:regexMonadicLessReg}
	\regexMonadicLessReg
\end{corollary}

\section{Extension to a Combined Relational/Textual
	Model}\label{sec:exten}

In this section, we extend Theorem~\ref{thm:ptime-generating} to 
\e{Spannerlog}, a data and query model introduced by Nahshon et
al.~\cite{DBLP:conf/webdb/NahshonPV16} that unifies and generalizes
relational databases and spanners by considering relations over both
strings and spans.

\subsection{Spannerlog}
The fragment of Spannerlog that we consider is referred to by Nahshon
et al.~\cite{DBLP:conf/webdb/NahshonPV16} as Spannerlog$\angs{\rgxc}$,
and we abbreviate it as simply $\sp$.  A \e{mixed signature} is a
collection of \e{mixed relation symbols} $R$ that have two types of
attributes: \e{string} attributes and \e{span} attributes. We denote
by $[R]\stri$ and $[R]\spni$ the sets of string attributes and span
attributes of $R$, respectively, where an attribute is represented by
its corresponding index.  Hence, $[R]\stri$ and $[R]\spni$ are
disjoint and $[R]\stri\cup [R]\spni=\set{1,\dots,\arity R}$.  A
\e{mixed relation} over $R$ is a set of tuples $(a_1,\dots,a_m)$ where
$m$ is the arity of $R$ and each $a_\ell$ is a string in $\Sigma^*$ if
$\ell\in[R]\stri$ and a span $\mspan i j$ if $\ell\in[R]\spni$.  A
\e{mixed instance} $D$ over a mixed signature consists of a mixed
relation $R^D$ for each mixed relation symbol $R$.  A \e{query}
$\query$ over a mixed signature $\E$ is associated with a mixed
relation symbol $R_Q$, and it maps every mixed instance $D$ over $\E$
into a mixed relation $Q(D)$ over $R_Q$.

A mixed signature whose attributes are all string attributes (in all
of the mixed relation symbols) is called a \e{span-free signature}.  A
mixed relation over a relation symbol whose attributes are all string
(respectively, span) attributes is called a \e{string relation}
(respectively, \e{span relation}).  To emphasize the difference
between mixed signatures 
(respectively, mixed relation symbols,mixed  relations)
and the signatures that do not involve types (which we have dealt with
up to this section), we often relate to the latter as \e{standard}
signatures (respectively, 
standard relation symbols, standard relations).

We consider queries defined by $\sp$ programs, which are defined as
follows. We assume two infinite and disjoint sets $\vars\stri$ and
$\vars\spni$ of \e{string variables} and \e{span variables},
respectively. To distinguish between the two, we mark a string
variable with an overline (e.g., $\ol{x}$).  By a \e{string term} we
refer to an expression of the form $\ol{x}$ or $\ol{x}_{y}$, where
$\ol{x}$ is a string variable and $y$ is a span variable.  In $\sp$,
an \e{atom} over an $m$-ary relation symbol $R$ is an expression of
the form $R(\tau_1,\dots,\tau_m)$ where $\tau_\ell$ is a string term
if $\ell\in[R]\stri$ or a span variable if $\ell\in[R]\spni$.  A
\e{regex atom} is an expression of the form $\angs{\tau}[\gamma]$
where $\tau$ is a string term and $\gamma$ is a regex formula.  Unlike
$\rgxlang$, in which there is a single input string, in $\sp$ a regex
atom $\angs{\tau}[\gamma]$ indicates that the input for $\gamma$ is
$\tau$.  We allow regex formulas to use only span variables. An
\e{$\sp$ program} is a quadruple $\angs{\E,\I,\Phi,\Rout}$ where:
\begin{itemize}
	\item $\E$ is a mixed signature referred to as the \e{EDB} signature;
	\item $\I$ is a mixed signature referred to as the \e{IDB} signature;
	\item $\Phi$ is a finite set of \e{rules} of the form $\varphi \la
	\psi_1,\ldots , \psi_m$ where $\varphi$ is an atom over $\I$ and
	each $\psi_i$ is an atom over $\I$, an atom over $\E$, or a regex
	atom;
	\item $\Rout\in\I$ is a designated \e{output} relation symbol.
\end{itemize}
We require the rules to be \emph{safe} in the following sense: \e{(a)}
every head variable occurs at least once in the body of the rule,
\e{and (b)} every string variable $\ol{x}$ 
in the rule
occurs, \e{as a string
	term}, in at least one relational atom 
(over $\E$ or $\I$) in the rule. 

We extend $\sp$ with \e{stratified negation}
in the usual way: the set of relation symbols in
$\E\cup\I$ is partitioned into \e{strata} $\I_0,\I_1,\dots,\I_m$ such
that $\I_0=\E$, the body of each rule contains only relation symbols
from strata that precede or the same as that of the head, and negated
atoms in the body are from strata that strictly precede that of the
head.  In this case, \emph{safe} rules are those for which every head
variable occurs at least once in a \e{positive} atom in the body of
the rule and every string variable $\ol{x}$ 
in the rule
occurs, as a string term,
in at least one \e{positive} relational atom 
(over $\E$ or $\I$) in the rule. 

The semantics of an $\sp$ program (with stratified negation) is
similar to the semantics of $\rgxlang$ programs (with the standard
interpretation of stratified negation in Datalog).  Given a mixed
instance $D$ over $\E$, the $\sp$ program $P=\angs{\E,\I,\Phi,\Rout}$
computes the mixed instance $P(D)$ over $\I$ and emits the mixed
relation $\Rout$ of $P(D)$. A query $Q$ over $\E$ is \e{definable} in
$\sp$ if there exists an $\sp$ program $P= \angs{\E,\I,\Phi,\Rout}$
such that $\Rout^{P(D)} = Q(D)$ for all mixed instances $D$ over $\E$.

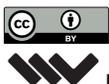
\begin{figure}[t]
	\centering
	\captionsetup[subfigure]{justification=centering}
	\hskip-0.6em
	\begin{subfigure}{4.3in}
		\small\centering
		$\rel{Geneo}$:\\
		\begin{tabular}{|l|}
			\hline
			\s{Cain}\visiblespace \s{son}\visiblespace \s{of} \visiblespace \s{Adam,}\visiblespace
			\s{Abel}\visiblespace\s{son}\visiblespace \s{of}\visiblespace \s{Adam,}\visiblespace \s{Enoch}\visiblespace \s{son}\visiblespace \s{of}\visiblespace \s{Cain,}\visiblespace \s{Irad}\visiblespace \s{...} \\
			\hline
			\s{Obed}\visiblespace \s{son}\visiblespace \s{of}\visiblespace \s{Ruth,}\visiblespace \s{Obed}\visiblespace \s{son}\visiblespace \s{of}\visiblespace \s{Boaz,}\visiblespace \s{Jesse}\visiblespace \s{son}\visiblespace \s{of}\visiblespace \s{Obed,}\visiblespace 
			\s{David}\visiblespace 
			\s{...}\\
			\hline
		\end{tabular}
	\end{subfigure}
	\caption{\label{fig:splog} The input for the program in Example~\ref{ex:splog}}
\end{figure}

\begin{example}
	\label{ex:splog}
	\def\cmsp{\,,\,}
	Following is an $\sp$ program 
	over the mixed signature of the instance of Figure~\ref{fig:splog}.
	As usual, $\Rout$ is the relation symbol in the head of the last rule, here
	$\rel{NotRltv}$.
	{\small
		\begin{align*}
		\rel{Ancstr}(\bar{x},y,\bar{x}, z) \dla&
		\rel{Geneo}(\bar{x}) \cmsp
		\angs{\bar{x}} \rgxname {prnt}(y,z)
		\\
		\rel{Ancstr}(\bar{w},y,\bar{x}, z) \dla & 
		\rel{Ancstr}(\bar{w},y,\bar{v}, y^\prime)\cmsp
		\rel{Geneo}(\bar{x})\cmsp
		\angs{\bar{x}} \rgxname {prnt}({z^\prime},z)\cmsp
		\rel{StrEq}(\bar{x}_{z^\prime}, \bar{v}_{y^\prime})
		\\
		\rel{Rltv}(\bar{w},y,\bar{x},z) \dla &
		\rel{Ancstr}(\bar{v},y^\prime,\bar{w},y) \cmsp
		\rel{Ancstr}(\bar{u},z^\prime,\bar{x},z,) \cmsp
		\rel{StrEq}(\bar{v}_{y^\prime}, \bar{u}_{z^\prime})
		\\
		\rel{NotRltv}(\bar{w},y,\bar{x},z) \dla &
		\rel{Geneo}(\bar{w})\cmsp
		\angs{\bar{w}} \rgxname {prsn}(y)\cmsp
		\rel{Geneo}(\bar{x})\cmsp
		\angs{\bar{x}} \rgxname {prsn}(z),\,
		\neg \rel{Rltv}(\bar{w},y,\bar{x},z)
		\end{align*}
	}
	The relation
	$\rel{Geneo}$ in Figure~\ref{fig:splog} contains strings that describe (partial) family trees. We assume
	for simplicity that every name that occur in such string is a unique identifier.
	The regex formulas $\rgxname{prsn}(x)$ and $\rgxname{prnt}(y,z)$ are the same as $\rgxname{cap}(x)$ and $\rgxname{prnt}(y,z)$ defined
	in  Example~\ref{ex:rgx}, respectively.              
	The first two rules of the program extract the relation
	$\rel{Ancstr}$ that has four attributes: the first and third are string attributes and the second and fourth are span attributes. 
	The first (respectively, third) attribute is the ``context'' string of the second (respectively, fourth) span attribute. 
	Observe the similarity to the corresponding
	definition in Example~\ref{ex:ancestor}. Here, unlike
	Example~\ref{ex:ancestor}, we need also to save the
	context string of each of the spans, and hence, we need two
	additional attributes.  
	The third rule uses the relation
	$\rel{StrEq}$ that holds tuples $(\bar{w},y,\bar{x},z)$
	such that the subtrings $\bar{w}_y$ and $\bar{x}_z$ are
	equal. This relation can be expressed in $\sp$ similarly to
	$\rgxlang$, as described in Section~\ref{sec:core}.  
	
	After evaluating the program, the relation $\rel{Ancstr}$
	holds tuples $(\bar{w}, y, \bar{x}, z)$ such that $\bar{w}_y$ is
	an ancestor of $\bar{x}_{z}$. 
	The relation
	$\rel{Rltv}$ holds tuples $(\bar{w}, y, \bar{x}, z)$ such that according to the information stored in $\rel{Geneo}$, 
	$\bar{w}_y$ is a relative of $\bar{x}_{z}$ (i.e., they
	share a common ancestor). 
	The relation $\rel{NonRltv}$ holds
	tuples $(\bar{w}, y, \bar{x}, z)$ such that $\bar{w}_y$ is \e{not} a
	relative of $\bar{x}_{z}$.\qed
\end{example}

\subsection{Equivalence to Polynomial Time}
\newcommand{\D}{\mathcal{D}}
\newcommand{\strfun}[1]{\textsf{str}^+(#1)}
\newcommand{\spnfun}[1]{\textsf{spn}^+(#1)}

Let $\E$ be a span-free signature, and $D$ an instance over
$\E$.  We define the \e{extended active domain} of $D$, in notation
$\adom^+(D)$, to be the union of the following two sets:
\e{(a)} the set of all strings that appear in $D$, as well as all of
their substrings;
\e{and (b)} the set of all spans of strings of
$D$.

Note that for every query $Q$ definable as an $\sp$ program $P=
\angs{\E,\I,\Phi,\Rout}$, and every input database $D$ over $\E$, we
have $\adom(Q(D)) \subseteq \adom^+(D)$, that is, every output string
is a substring of some string in $D$, and every output span is a span
of some string in $D$. Our result in this section states that, under
this condition, we can express in $\sp$ with stratified negation every
query $Q$, as long as $Q$ is computable in polynomial time.

\def\splogPtime { Let $Q$ be a query over a span-free signature $\E$,
	with the property that
	$\adom(Q(D)) \subseteq \adom^{+}(D)$ for all
	instances $D$ over $\E$. 
	The following are equivalent:
	\begin{enumerate}
		\item $Q$ is computable in polynomial time.
		\item $Q$ is computable in $\sp$ with stratified negation.                      
	\end{enumerate} 
}
\begin{theorem}\label{thm:ptime-sp-decision}
	\splogPtime
\end{theorem}

We remark that Theorem~\ref{thm:ptime-sp-decision} can be extended to
general mixed signatures $\E$ if we assume that every span mentioned
in the input database $D$ is within the boundary of some string in
$D$. Moreover, Theorem~\ref{thm:ptime-sp-decision} is not 
correct without stratified negation, and this can be shown using
standard arguments of monotonicity.

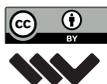
\partitle{Proof idea}
\label{sec:ProofExt}
We now discuss the proof idea of Theorem~\ref{thm:ptime-sp-decision}.
The full proof is in the appendix. 
The direction $2\ra 1$ is straightforward, so we
discuss only
the direction $1\ra
2$.  Let $Q$ be a query over a span-free signature $\E$, with the
property that $\adom(Q(D)) \subseteq \adom^{+}(D)$ for all instances
$D$ over $\E$. Assume that $Q$ is computable in polynomial time. 
We need to construct an $\sp$ program $P$ with stratified negation for
computing $Q$.  
We do so
in two steps.  In the first step, we
apply Theorem~\ref{thm:ptime-generating} to get a (standard)
$\semiposdatalog$ program $P'$ that simulates $Q$.  
Yet,
$P'$ does not necessarily respect the \e{typing}
conditions of $\sp$ with respect to the two types \e{string} and
\e{span}. So, in the second step, we transform $P'$ to an $\sp$
program $P$ as desired. Next, we discuss each step in more detail.

\partitle{First step}
In order to produce the $\semiposdatalog$ program $P'$, some
adaptation is required to apply
Theorem~\ref{thm:ptime-generating}. First, we need to deal with the
fact that the output of $Q$ may include values that are not in the
active domain of the input (namely, spans and substrings).  Second, we
need to establish a linear order over the active domain. Third, we
need to assure that the query that Theorem~\ref{thm:ptime-generating}
is applied to respects isomorphism. To solve the first problem, we
extend the input database $D$ with relations that contain every
substring and every span of every string in $D$. This can be done
using $\sp$ rules with regex atoms. For the second problem, we
construct a linear order over the domain of all substrings and spans
of strings of $D$, again using $\sp$ rules. For this part, stratified
negation is needed.  For the third problem, we show how our extended
input database allows
us
to restore $D$ even if all values (strings and
spans) are replaced with other values by applying an injective
mapping.

\partitle{Second step}
In order to transform $P'$ into a ``legal'' $\sp$ program $P$ that
obeys the typing of attributes and variables, we do the
following. First, we replace every IDB relation symbol $R$ with every
possible \e{typed} version of $R$ by assigning types to
attributes. Semantically, we view the original $R$ as the union of all
of its typed versions. Second, we replace every rule with every typed
version of the rule by replacing relation symbols with their typed
versions. Third, we eliminate rules that treat one or more variable
inconsistently, that is, the same variable is treated once as a string
variable and once as a span variable.  Fourth, we prove that this
replacement preserves the semantics of the program.

\section{Conclusions}\label{sec:conclusions}
We studied $\rgxlang$, namely, Datalog over regex formulas. We proved
that this language expresses precisely the spanners that are
computable in polynomial time. $\rgxlang$ is more expressive than the
previously studied language of core spanners and, as we showed here,
more expressive than even
the language of generalized core spanners. We also observed that it
takes very simple binary regex formulas to capture the entire
expressive power. Unary regex formulas, on the other hand, do not 
suffice:  in the Boolean case, they recognize precisely the regular languages, 
and in the non-Boolean case, they produce a strict subset of the regular spanners. 
Finally, we extended the equivalence result to
$\sp$ with stratified negation over mixed instances, a model that
generalizes both the relational model and the document spanners.

The remarkable expressive power of $\rgxlang$ is somewhat mysterious,
since we do not yet have a good understanding of how to phrase some
simple polynomial-time programs \e{naturally} in $\rgxlang$. The
constructive proof simulates the corresponding polynomial-time Turing
machine, and does not lend itself to program clarity.  For instance,
is there a natural program for computing the \e{complement} of the
transitive closure of a  binary relation encoded by the 
input?
In future work we plan to investigate this aspect by studying the
complexity of translating simple formalisms, such as generalized core
spanners, into $\rgxlang$.

{\bibliographystyle{abbrv}
\bibliography{main}}

\appendix
\def\pagenref{\ref{thm:pa-gen-core}}
\section{Proof of Theorem~\pagenref}
In this section, we prove Theorem~\ref{thm:pa-gen-core}.

\begin{reptheorem}{\ref{thm:pa-gen-core}}
	\thmpagencore
\end{reptheorem}

We begin with the ``only if'' direction, which is the more involved
direction.
This is the direction we are most interested in, since it gives us Theorem~\ref{thm:pow2}.

\subsection{The ``Only If'' Direction}
Let $\gamma$ be an expression in
$\rgxc\alg{\cup,\pi,\join,\sel^=,\setminus}$. We say that $\gamma$ is
\e{positional} if for every expression $x\set{\delta}$ that occurs in
$\gamma$
it is the case that $\delta=\epsilon$. Therefore, all span
variables are assigned empty spans, and hence, represent positions in
the input string $\strs$. Observe that if $\gamma$ is positional, then
$\sel^=$ is redundant since every two spans 
referenced by $\gamma$
have the same string. To
provide positional expressions with the needed expressive power, we
also add the selection $\sel^=_-$ that takes as input four spans
$x_1$, $x_2$, $y_1$ and $y_2$ and returns
$\true$
if 
$x_1$ precedes $x_2$, and $y_1$ precedes $y_2$, and
the string between
$x_1$ and $x_2$ is equal to the string between $y_1$ and $y_2$.

Each expression $\delta$ in  $\rgxc\alg{\cup,\pi,\join,\sel^=_-,\setminus}$  is built out of applying operators 
to regex formulas: we call these regex formula \e{components} of $\delta$.
A direct consequence of the translation of regex formulas into \e{path
	unions}~\cite{DBLP:journals/jacm/FaginKRV15}
shows the following.
\begin{lemma}\label{lemma:gen-core-positional}
	For every Boolean $\gamma$ in
	$\rgxc\alg{\cup,\pi,\join,\sel^=,\setminus}$ there exists a
	positional Boolean $\delta$ in
	$\rgxc\alg{\cup,\pi,\join,\sel^=_-,\setminus}$ such that all of the
	following hold.
	\begin{enumerate}
		\item Each 
		component of
		$\delta$ is of the form $[x\set{}.^*]$, or
		$\angs{x\set{}\alpha y\set{}}$, or $[.^*x\set{}]$, where $\alpha$ is
		a (variable-free) regular expression.
		\item $\rep{\delta}=\rep{\gamma}$; that is, $\delta$ and $\gamma$
		accept the same strings.
	\end{enumerate}
\end{lemma}
Note that in the first condition of
Lemma~\ref{lemma:gen-core-positional}, the regex formula $[x\set{}.^*]$
states that $x$ is the first position, $[.^*x\set{}]$ states that $x$
is the last position, and $\angs{x\set{}\alpha y\set{}}$ states that
the string between $x$ and $y$ satisfies the regular expresion $\alpha$.

By a slight abuse of notation, we view a variable $x$ in a positional
formula $\gamma$ as a natural number that represents its location.
For example, if $x$ is assigned the span $\mspan{5}{5}$ then we view
$x$ simply as $5$.  Our central lemma is the following.

\begin{lemma}\label{lemma:balder}
	Let $\delta$ be a positional expression as in
	Lemma~\ref{lemma:gen-core-positional}. Then $\delta$ is equivalent
	to a disjunction of formulas of the form
	$\omega(x_1,\dots,x_n)\land \varphi(z_0,\dots,z_n)$
	where:
	\begin{enumerate}
		\item $\omega(x_1,\dots,x_n)$ specifies a total order $x_{i_1}\leq
		x_{i_2}\leq\dots\leq x_{i_n}$ over $x_1,\dots,x_n$ (viewed as
		numeric positions).
		\item $\varphi(z_0,\dots,z_n)$ is a PA formula, where each $z_i$
		represents the length $i$th segment (among the $n+1$ segments) of
		$\strs$ as defined in $\omega(x_1,\dots,x_n)$.
	\end{enumerate}
\end{lemma}
\begin{proof}
	We prove the 
	lemma
	by induction on the structure of $\delta$.
	We first handle the case where $\delta$ is 
	atomic, that is, one of the three forms of components in part~(1) of Lemma~\ref{lemma:gen-core-positional}.
	
	\begin{itemize}
		\item If $\delta$ is $[x\set{\epsilon}.^*]$, then it is equivalent to
		$\true \land z_0=0$.
		\item If $\delta$ is $\angs{x\set{ \epsilon}\alpha y\set{\epsilon}}$, then it is
		equivalent to
		$(x\leq y) \land \varphi(z_1)$ where $\varphi(z_1)$ is the PA
		formula stating that $z_1$ is a length of a string in $\a^*$
		satisfying
		$\alpha$. 
		It is known
		that such $\varphi$ exists,
		since $\alpha$ is a regular expression~\cite{ginsburg1966semigroups,parikh1966context}.
		\item If $\delta$ is $[.^*x\set{\epsilon}]$, then it is equivalent to
		$\true \land z_1=0$.
	\end{itemize}
	
	Next, we consider algebraic expressions and use the induction
	hypothesis. We assume that $\delta(x_1,\dots,x_n)$ is equivalent to
	$$\bigvee_{i=1}^k
	\omega_i(x_1,\dots,x_n)\land
	\varphi_i(z_{i,0},\dots,z_{i,n})
	$$
	and that $\delta'(x'_1,\dots,x'_\ell)$ is equivalent to
	$$\bigvee_{i=1}^\ell
	\omega'_i(x'_1,\dots,x'_\ell)\land
	\varphi'_i(z'_{i,0},\dots,z'_{i,\ell})\,.
	$$
	
	For $\delta_1\cup\delta_2$ we assume union compatibility, which means
	that $\set{x_1,\dots,x_n}=\set{x'_1,\dots,x'_\ell}$. Hence, we simply
	take the 
	disjunction of the two disjunctions.
	
	For $\pi_{y_1,\dots,y_q}(\delta)$, we replace each
	$\omega_i(x_1,\dots,x_n)$ with $\omega_i(y_1,\dots,y_q)$ by simply
	restricting the total order to  $y_1,\dots,y_q$. In addition,
	we replace each $\varphi_i(z_{i,0},\dots,z_{i,n})$ 
	with the PA formula
	$$
	\exists {z_{i,0},\dots,z_{i,n}}\Big[\varphi_i(z_{i,0},\dots,z_{i,n})
	\land\xi(z_0,\dots,z_q,z_{i,0},\dots,z_{i,n})\Big]
	$$
	where $\xi(z_0,\dots,z_q,z_{i,0},\dots,z_{i,n})$ states the
	relationships between the lengths $z_0,\dots,z_q$ and
	$z_{i,0},\dots,z_{i,n}$, 
	stating that each $z_i$ is the sum of some
	of variables from $z_{i,0},\dots,z_{i,n}$. For example, if
	$\omega_i(x_1,\dots,x_n)$ is $x_1\leq x_2\leq x_3$ and 
	the operation is $\pi_{x_1,x_3}\delta$, then 
	$\xi$ will be
	\[z_0=z_{i,0} \land z_1=z_{i,1}+z_{i,2} \land z_2=z_{1,3}\,.\] 
	
	For $\sel^=_{x_1,x_2,x_3,x_4})\delta$, 
	we replace each
	$\varphi_i(z_{i,0},\dots,z_{i,n})$ with the conjunction
	$\varphi_i(z_{i,0},\dots,z_{i,n})\land\xi(z_{i,0},\dots,z_{i,n})$
	where $\xi(z_{i,0},\dots,z_{i,n})$ states the equality on the sum of
	corresponding segments expressed by the selection condition 
	$x_2 -x_1 = x_4 - x_3$.
	
	We are left with natural join ($\join$) and difference
	($\setminus$). For that, we will show how to express both
	$\delta\land\delta'$ and 
	$\neg\delta$
	in the form of the lemma.
	
	For $\delta\land\delta'$, we transform the conjunction of each pair of
	disjuncts (one from $\delta$ and one from $\delta'$) separately.
	To represent the conjunction of 
	$\omega_i(x_1,\dots,x_n)\land
	\varphi_i(z_{i,0},\dots,z_{i,n})$
	and 
	$\omega'_j(x'_1,\dots,x'_\ell)\land
	\varphi'_j(z'_{j,0},\dots,z'_{j,\ell})$
	we take the disjunction over all total orders
	$$\omega(x_1,\dots,x_n,x'_1,\dots,x'_\ell)$$
	obtained by interpolating the two total orders (hence, preserving each
	order separately). For each such interpolated order,
	we represent the conjunction
	$$ \varphi_i(z_{i,0},\dots,z_{i,n})\land
	\varphi'_j(z'_{j,0},\dots,z'_{j,\ell})$$ by replacing each segment
	variable $z$ with the corresponding sum of segments from the interpolated order
	$\omega(x_1,\dots,x_n,x'_1,\dots,x'_\ell)$.
	
	Finally, for 
	$\neg\delta$,
	we take the disjunction over all total
	orders $\omega(x_1,\dots,x_n)$ 
	of:
	$$\omega(x_1,\dots,x_n)\land \bigwedge_{i=1}^k \psi_i(z_0,\dots,z_n)$$
	where each $\psi_i(z_0,\dots,z_n)$ is the disjunction of the following
	two:
	\begin{itemize}
		\item The total order $\omega(x_1,\dots,x_n)$ is \e{incompatible} with
		the order $\omega_i(x_1,\dots,x_n)$, which means that a segment that
		should be empty (e.g., since $x_1\leq x_3$ in $\omega$ and $x_3\leq
		x_1$ in $\omega_i$) has length larger than zero. Hence,
		$\omega(x_1,\dots,x_n)$ is a disjunction of such statements, each
		handling one segment.
		\item The total order $\omega(x_1,\dots,x_n)$ is \e{compatible} with
		the order $\omega_i(x_1,\dots,x_n)$, which means that some segments
		should be empty and $\neg\varphi_i(z_{i,0},\dots,z_{i,n})$.
	\end{itemize}
	This completes the proof.
\end{proof}
By applying Lemma~\ref{lemma:balder} to a Boolean $\delta$ we get the following.
\begin{lemma}\label{lemma:balder-boolean}
	If $\delta$ is a Boolean positional expression as in
	Lemma~\ref{lemma:gen-core-positional}, then the language recognized
	by $\rep{\delta}$ is definable in PA.
\end{lemma}
Finally, combining Lemma~\ref{lemma:gen-core-positional} with
Lemma~\ref{lemma:balder-boolean} we conclude ``only if'' direction of
Theorem~\ref{thm:pa-gen-core}, as required.

\subsection{The ``If'' Direction}
Let $\varphi(x)$ be a unary PA formla. For the ``if'' direction we
need to show the existence of a Boolean generalized core spanner
$\delta$ that recognizes a language $L\subseteq\set{\a}^*$ such that
$\nat(L)$ is the set of natural numbers defined by $\varphi(x)$; that
is, for all strings $\strs\in\a^*$ it is the case that
$\rep{\delta}(\strs)=\true$ if and only if $\varphi(|\strs|)$.  This
direction of Theorem~\ref{thm:pa-gen-core} is simpler, due to a key
result by Presburger~\cite{Presburger29} who proved that PA admits
\e{quantifier elimination} (cf.~\cite{Enderton:2001} for a modern
exposition). 
We make use of
the following theorem.
\begin{theorem}\label{thm:balder-reverse}
	\quad Let $\varphi(x_1,...,x_k)$ be a PA formula. There is a formula
	$\gamma\in\rgxc\alg{\cup,\pi,\join,\sel^=,\setminus}$ with
	$\vars(\gamma)=\set{w_1,\dots,w_k}$ such that for all
	$\strs\in\a^*$, the following are equivalent for all records
	$r:\vars(\gamma)\ra\allspans(\str s)$:
	\begin{enumerate}
		\item $r\in\rep{\gamma}(\strs)$;
		\item $\varphi(|\strs_{r(w_1)}|,\dots,|\strs_{r(w_k)}|)$.
	\end{enumerate}
\end{theorem}
\begin{proof}
	Presburger~\cite{Presburger29} proved that every formula in PA is
	equivalent to a quantifier-free formula built up from the following
	symbols.
	\begin{itemize}
		\item The constants $0$ and $1$;
		\item The $+$ function; 
		\item The binary predicate $<$;
		\item The unary \e{divisibility} predicate $\equiv_k$, for all $k\in
		\mathbb{N}$, where ${\equiv_k}(x)$ is interpreted as ``$x$ is
		divisible by $k$.''
	\end{itemize}
	
	We first eliminate the use of complex terms at the cost of
	reintroducing quantifiers, but only of a particular, bounded form: by
	``bounded existential quantification'' we mean existential
	quantification of the form $\exists y [y<x \land \ldots]$, or written
	as $\exists y<x(\ldots)$ for short, where $x$ and $y$ are distinct
	variables. It follows from Pressburger~\cite{Presburger29} that every
	PA formula can be equivalently written as a formula built up from
	atomic formulas of the form $x=0$, $x=1$, $x=y$, and $x=y+z$, using
	the Boolean connectives and bounded existential quantification. In
	particular, $x<y$ can be expressed as $\exists z<y [x=z]$ and
	$\equiv_k(x)$ can be expressed as
	\begin{align*}
	x=0 &\,\, \bigvee\,\,
	\exists y_1<x\,\exists y_2<x \, \ldots \exists y_{k-1}<x
	\left[\left( \bigwedge_{i=2,\dots,k-1} (y_i=y_{i-1} + y_1)\right)\land x=y_{k-1}+y_1
	\right]\,.
	\end{align*}
	This formula says that there is $y_1$ such that $x = y_1 k$.
	
	So, assuming that $\varphi$ has the above structure, we continue the
	proof by induction. For the basis we have the following:
	\begin{itemize}
		\item For $x_1=0$ we use $\angs{w_1\set{\epsilon}}$.
		\item For $x_1=1$ we use $\angs{w_1\set{.}}$.
		\item For $x_1=x_2$ we use
		$\sel^=_{w_1,w_2} (\angs{w_1\set{.^*}}\join\angs{w_2\set{.^*}})$.
		\item For $x_1=x_2+x_3$ we use
		$$
		\pi_{w_1,w_2,w_3}
		\sel^=_{w_2,w_2'}\sel^=_{w_3,w_3'}
		\angs{w_1\set{w_2'\set{.^*}w_3'\set{.^*}}}\join\angs{w_2\set{.^*}}\join\angs{w_3\set{.^*}}
		$$
	\end{itemize}
	For the inductive step, we need to show closure under conjunction,
	negation, and bounded existential quantification. Let
	$\varphi(x_1,\dots,x_k)$ and $\varphi'(x'_1,\dots,x'_\ell)$ be two PA
	formulas, and $\gamma(w_1,\dots,w_k)$ and
	$\gamma'(w'_1,\dots,w'_\ell)$ the corresponding expressions in
	$\rgxc\alg{\cup,\pi,\join,\sel^=,\setminus}$. To express
	$\varphi\land\varphi'$ we simply use $\gamma\join\gamma'$.
	To express $\neg\gamma$ we use 
	$\left(\angs{w_1\set{.^*}}\join\dots\join\angs{w_k\set{.^*}}\right)\setminus\gamma$.
	Finally, for the formula $\exists x_1<x_2\,\varphi(x_1,\dots,x_k)$ we use
	$\pi_{w_2,\dots,w_{k}}
	\sel^=_{w_1,w_1'}
	\angs{w_2\set{.^+\,w_1'\set{.^*}}}\join\gamma$.
	This completes the proof.
\end{proof}
From Theorem~\ref{thm:balder-reverse} we conclude the ``if'' direction
of Theorem~\ref{thm:pa-gen-core}, as follows. Let $\varphi(x)$ be a
unary PA formla. Let $\gamma(w)$ be the corresponding formula of
Theorem~\ref{thm:balder-reverse}. We define a Boolean $\delta$ in
$\rgxc\alg{\cup,\pi,\join,\sel^=,\setminus}$ as follows.
$$\delta \mathrel{{:}{=}} \pi_{\emptyset}\left([w\set{.^*}]\join \gamma(w)\right)$$
From Theorem~\ref{thm:balder-reverse} we conclude that $\delta(\strs)$
is true if and only if $\varphi(|\strs|)$, as required.

\section{Proof of Theorem~\ref{thm:monadic}}

\begin{reptheorem}{\ref{thm:monadic}}
	\monadicRecRelation
\end{reptheorem}
\begin{proof}
	We first show the direction $1\ra 2$.  Suppose that $S$ is definable
	by the regex-monadic program $P$.  We can assume, without loss of
	generality, that every regex formula $\gamma(x)$ in $P$ appears only
	in a rule of the form $R_\gamma(x) \leftarrow \gamma(x)$ where $x$
	is the variable of $\gamma$.  The result of running the $\rgxlang$
	program $P$ is then the same as the result of running the ordinary
	Datalog program $P^\prime$ over an ordinary relational database
	where the
	relations in the EDB
	are the relations $R_{\gamma}$, which are populated
	by the above rules.  Since the 
	relations in the EDB
	of $P^\prime$ are unary, it
	follows from Levy et al.~\cite{DBLP:conf/pods/LevyMSS93} that
	$P^\prime$ is equivalent to a nonrecursive Datalog program
	$P^{\prime \prime}$.  In turn, the nonrecursive $P''$ is equivalent
	to a union of conjunctive queries, and hence, we conclude that $S$
	is equivalent to a $\rgxlang$ program where all the rules have
	the form
	$$
	\Rout(x_1,\ldots,x_k) \leftarrow \beta_1(y_{1}), \ldots, \beta_m(y_m)
	$$
	where $\set{x_i,\dots,x_k}\subseteq \set{y_1,\dots,y_m}$ and each
	$\beta_i(y_i)$ is a unary regex formula.
	
	Since the conjunction of regex formulas is a regex
	formula~\cite{DBLP:journals/jacm/FaginKRV15}, we can group together
	regex formulas that have the same variable, and therefore we can
	assume that the $y_i$ are unique (that is, $y_i\neq y_j$ whenever
	$i\neq j$).  If there is a variable $y_{i}$ that is not in
	$\set{x_i,\dots,x_k}$, then we remove $y_i$ from the regex formula
	$\beta_i$. This removal does not affect the semantics of the rule,
	since $y_i$ occurs only once.  We can then compile all of the Boolean
	regex formulas that result from the removal of the $y_i$ into one
	Boolean regex formula, and take it as our $\gamma()$.  If there are no
	such $y_i$ (i.e., all body variables occur in the head), then we can
	define $\gamma()$ vacuously as $[ .^* ]$. Finally, to construct the
	formula as in the theorem, we take as $\gamma_i(x_i)$ the atom
	$\beta_j(y_j)$ where $y_j=x_i$. 
	This may require duplicating an atom
	(with no semantic impact) if a head variable occurs more than once,
	that is, $x_i=x_j$ for some $i\neq j$.  
	For example, the rule 
	$\Rout(x,x) \leftarrow \beta(x),\gamma()$
	becomes
	$\Rout(x,x) \leftarrow \beta(x),\beta(x),\gamma()$,
	where redundancy is added to match the form of the theorem.  
	
	The direction $2\ra 1$ is straightforward, since the form of Part~2
	is ``almost'' regex-monadic. Indeed, while the regex formula $\gamma()$ is
	of arity zero and not one, we can simply add a dummy variable to it,
	say $x_0$. For example, if $\gamma()$ is the regex formula $[\alpha]$
	for a regular expression $\alpha$, then we can replace $\gamma()$
	with $\gamma_0(x_0)=[x_0\set{}\cdot\alpha]$ or $\gamma_0(x_0)=[x_0\set{\alpha}]$.
\end{proof}

\section{Proof of Corollary~\ref{cor:regexMonadicLessReg}}

To prove Corollary~\ref{cor:regexMonadicLessReg}, we use the following
lemma.

\begin{lemma}\label{lemma:equiv-K}
	Let $P$ be a regex-monadic program. There is a constant natural
	number $K$ such that for all input strings $\strs$ there is an
	equivalence relation on spans, with at most $K$ equivalence classes,
	such that following holds. Every output record in $\rep{P}(\strs)$
	remains an output record whenever a span is replaced with 
	a span in the same equivalence class.
\end{lemma}
\begin{proof}
	Assume that $P$ has the form of the second part of
	Theorem~\ref{thm:monadic}. We take as our equivalence 
	relation
	the relation
	$x\equiv y$
	stating that $x$ and $y$ are produced by the exact same set of regex
	formulas $\gamma_i(x_i)$ of the rules in $P$. The number of
	equivalence classes is then bounded by the number of sets of atoms
	in $P$.
\end{proof}

Corollary~\ref{cor:regexMonadicLessReg} follows from
Theorem~\ref{thm:monadic} and Lemma~\ref{lemma:equiv-K}, as we show
next. 

\begin{repcorollary}{\ref{cor:regexMonadicLessReg}}
	\regexMonadicLessReg
\end{repcorollary}
\begin{proof}
	Since every program in the form of the second part of
	Theorem~\ref{thm:monadic} is the union of joins of regex formulas,
	we get from known results~\cite{DBLP:journals/jacm/FaginKRV15} that
	every regex-monadic program defines a regular spanner.  To show that
	the expressive power is strictly smaller, we will show that
	containment of spans
	cannot be expressed by a regex-monadic program.
	Formally, let $S$ be the
	spanner $\repspnr{\gamma}$ where $\gamma$ is $\langle x \{ .^* y \{
	.^* \} .^* \} \rangle $, that is, $\gamma$ extracts all pairs $x$
	and $y$ of spans such $x$ contains $y$.  We will
	show that $S$ is not equivalent to any monadic $\rgxlang$ program.

	Assume, by way of contradiction, that $S$ is equivalent to the
	monadic $\rgxlang$ program $P$, and let $K$ be the number in
	Lemma~\ref{lemma:equiv-K}. We assume that $K$ is large enough so
	that the number of spans in a string of length $K$ is larger than
	$K$. We can make this assumption, since the number of spans of a
	string of length $n$ is $\Theta(n^2)$.  Take $\strs$ to be a string
	of length $K$.  Then $\strs$ has more than $K$ spans, so some
	equivalence class has at least two distinct spans $s$ and $t$ of
	$\strs$.  Note that for every two distinct spans $s$ and $t$ there
	exists a span $u$ such that either \e{(a)} the span $u$ contains $s$
	but does not contain $t$, \e{or (b)} the span $u$ contains $t$ but
	does not contain $s$. From Lemma~\ref{lemma:equiv-K} it follows that
	if $R(u,s)$ holds if and only if $R(u,t)$ also holds, hence a
	contradiction.
\end{proof}

\def\pagenref{\ref{thm:ptime-sp-decision}}
\section{Proof of Theorem~\pagenref}

In this section, we prove Theorem~\ref{thm:ptime-sp-decision}. We
first restate the theorem.
\begin{reptheorem}{\ref{thm:ptime-sp-decision}}
	\splogPtime
\end{reptheorem}

\newcommand{\inputinst}{D} 

Throughout this section, we fix a span-free signature $\E$ as our
input signature, and a query $Q$ over $\E$.  We will prove that if $Q$
is computable in polynomial time, then it can be phrased as an $\sp$
program with stratified negation.
The other direction, that every $\sp$
program with stratified negation
can be
executed in polynomial time, is straightforward, similarly to ordinary
Datalog.  As described in Section~\ref{sec:exten}, our proof comprises
two steps, which we now construct.

\subsubsection*{First Step}
We first extend $\E$ with additional relation symbols. We implicitly
assume that each added relation symbol does not already belong to
$\E$.

\partitlelight{\underline{$\mathcal{R}^{\mathsf{type}}$:}}\\
The mixed signature $\mathcal{R}^{\mathsf{type}}$ consists of the
unary string relation $\Str$ and the unary span relation $\Spn$.

\partitlelight{\underline{$\mathcal{R}^{\Sigma}$:}}\\
The mixed signature $\mathcal{R}^{\Sigma}$ consists of the relation
symbols $R_{\sigma}$ for all $\sigma \in \Sigma$, where
$\arity{R_{\sigma}}=2$, the first attribute is a string attribute
(i.e., $1\in [R_{\sigma}]\stri$) and the second attribute is a span
attribute (i.e., $2\in [R_{\sigma}]\spni$).

\partitlelight{\underline{$\mathcal{R}^{\mathsf{ord}}$:}}\\
The mixed signature $\mathcal{R}^{\mathsf{ord}}$ that consists of the
following relation symbols:
\begin{itemize}
	\item
	$\first$ is a
	string relation with arity $1$;
	\item
	$\succrel\stri$ is a string relation with arity $2$;
	\item
	$\succrel\mix$ is a mixed relation with arity $2$ and $1\in[\succrel\mix]\spni$ and $2\in[\succrel\mix]\stri$;
	\item
	$\succrel\spni$ is a span relation with arity $2$;
	\item
	$\last$ is a
	span relation with arity $1$.
\end{itemize}

We denote by $\sigordplus$  
the signature $\E \cup
\mathcal{R}^{\mathsf{type}} \cup \mathcal{R}^{\Sigma} \cup \sigord$.
A mixed instance $E$ over $\sigordplus$ is said to \e{encode} an
instance $D$ over $\E$ if all of the following conditions hold.

\begin{enumerate}
	\item $R^{E} = R^D$ for  all $R\in \E$.
	\item The unary string relation $\rel{Str}^{E}$ consists of all
	strings in $\adom^+(D)$.
	\item The unary span relation $\rel{Spn}^{E}$
	consists of all of the spans in $\adom^+(D)$.
	\item Each $R_{\sigma}^{E}$ consists of the tuple $(\ol{x},y)$ where
	$\ol{x}$ is a string that occurs in $D$, and $y$ is a span of
	$\ol{x}$ of length one with $\ol{x}_y = \sigma$.
	\item The relations of $E$ that instantiate the signature 
	$\mathcal{R}^{\mathsf{ord}}$ interpret this signature so that
	the union $\succrel\stri^E \cup \succrel\mix^E \cup \succrel\spni^E$
	is a successor relation of a linear order over $\adom^+(D)$, wherein
	all strings precede all spans, and $\first$ and $\last$ determine the
	first and last elements in this linear order, respectively.
\end{enumerate}

Note the following in the last item above.  Since the strings precede
the spans in the linear order, the relation symbol $\first$ is a unary
string relation and $\last$ is a unary span relation.
The relation $\succrel\mix^E$ contains exactly one tuple $(\bar{x}, y)$, where $\bar{x}$ is the last string and $y$ is the first span.

We denote by $\spencode{\inputinst}$ the mixed instance over
$\sigordplus$ that encodes $D$.  A \e{mixed encoding} of an instance
$D$ over a span-free signature $\E$ is a mixed instance over
$\sigordplus$ that is isomorphic to $\spencode{D}$.  We define the
\e{untyped encoding} of a mixed encoding $D^{\prime\prime}$ to be the
instance obtained from $D^{\prime\prime}$ by viewing it as an instance
over the signature $\untyped{\sigordplus}$, where
$\untyped{\sigordplus}$ is an ordinary signature obtained from
$\sigordplus$ by \e{(a)} ignoring the types, \e{and (b)} 
relating to
the relation symbols $\succrel\spni$, $\succrel\stri$ and
$\succrel\mix$ uniformly as the binary successor relation symbol $\succrel$.

Note that a mixed encoding has a unique untyped encoding, and vice
versa. In particular, for every untyped encoding $D^\prime$ there
exists a unique mixed encoding $D^{\prime\prime}$ such that $D^\prime$
is the untyped encoding of $D^{\prime\prime}$. This is true, since we
can distinguish between spans and strings via the relations
$\Str^{D^\prime}$ and $\Spn^{D^\prime}$.

Similarly to Lemma~\ref{lem:DEncodesS}, we have the following. 
\begin{lemma}\label{lem:spanExtendedEncodingRestoreI}
	Let $D^\prime$ be an instance over $\untyped{\sigordplus}$.
	The following hold:
	\begin{enumerate}
		\item 
		Whether $D^\prime$ is an untyped encoding can be determined in polynomial time.
		\item 
		If $D^\prime$ is an untyped encoding, then there is a unique instance $D$ over $\E$
		and isomorphism $\iota$ such that
		$\iota(\spencode{D}) = D^\prime$; moreover, both $D$ and
		$\iota$ are computable in polynomial time.
	\end{enumerate}
\end{lemma}
\begin{proof}
	Note that an untyped encoding of an instance $D$ encodes each
	(string) entry of $D$ using relations over the mixed signature
	$\mathcal{R}^\Sigma$. Unlike
	Lemma~\ref{lem:DEncodesS},
	where we had a single
	string to encode, here we have a database of strings. Therefore, the
	mixed relations $R_{\sigma}$ holds an additional string attribute
	that indicates which entry in $D$ is encoded by the tuple.  The rest
	of this proof is a straightforward adaptation of that of
	Lemma~\ref{lem:DEncodesS}.
\end{proof}
\newcommand{\spencodextord}[1]{\mathsf{Enc}^+_{\mathsf{ut}}({#1})}

Let ${Q}$ be a query over a span-free signature $\E$.  We define the
query $Q^+$ over $\untyped{\sigordplus}$ on an input $D_1$ in the
following way: If $D_1$ is an untyped encoding of $\inputinst$ over
$\E$ then ${Q^+}(D_1) = \iota({Q}(\inputinst))$ where $\iota$ is as in
Lemma~\ref{lem:spanExtendedEncodingRestoreI}; otherwise ${Q^+}(D_1)$
is empty.  To apply 
Theorem~\ref{thm:papagen} 
on $Q^+$, we make the
following observation based on the definition of an untyped encoding
and on 
Lemma~\ref{lem:spanExtendedEncodingRestoreI}.

\begin{observation}\label{lem:Q1isPoly}
	The  query
	${Q^+}$ respects isomorphism, and moreover, is computable in polynomial time whenever
	${Q}$ is computable in polynomial time.
\end{observation}
Note also that the query $Q^+$ is defined over an ordered (standard)
signature due to the relations $\succrel$, $\first$ and $\last$.  Due
to this and to Observation~\ref{lem:Q1isPoly}, we can now apply
Theorem~\ref{thm:papagen} on $Q^+$ and obtain the following.
\begin{lemma}\label{lem:Q3P_to_DL}
	If $Q$ is computable in polynomial time, then there exists a
	$\semiposdatalog$ program $P_{Q}$ over $\untyped{\sigordplus} $ such
	that for every instance $D$ over $\E$ and every untyped encoding
	$D^\prime$ of $D$ it holds that $P_{Q}(D^\prime)$ is equal to
	$Q^+(D^\prime)$.
\end{lemma}

This completes the first step, where we translate $Q$ into an ordinary
program $P_Q$ over an ordinary signature. In the next step, we
transform $P_Q$ into an $\sp$ program over $\E$.

\subsubsection*{Second Step}
Due to the syntactic resemblance between Datalog and $\sp$, one could
suggest to consider Datalog rules over an ordinary signature simply as
$\sp$ rules.  However, there is a difference between the semantics of
the languages since Datalog programs get standard input instances, as
opposed to $\sp$ programs that get mixed instances and distinguish
between types.  We prove the following.
\begin{lemma}~\label{lem:Datalog_to_SP} If $Q$ is computable in
	polynomial time, then there exists an $\sp$ program $P^\prime_Q$
	over $\sigordplus$ such that for every instance $D$ over $\E$ and every mixed
	encoding $D^{\prime\prime}$ of $D$ it holds that
	$P^\prime_Q(D^{\prime\prime})$ equals $ Q^+(D^{\prime})$ where
	$D^{\prime}$ is the untyped encoding of $D^{\prime\prime}$.
\end{lemma}
\begin{proof}
	Due to Lemma~\ref{lem:Q3P_to_DL} it suffices to show how to
	translate $P_Q$ to a $\sp$ program.  A \e{mixed version} $\rho^+$ of
	a $\semiposdatalog$ rule $\rho $ is obtained by replacing each
	relation atom $R(x_1,\ldots,x_k)$ that appears in $\rho$ by all of
	the atoms obtained from it by assigning its attributes all of the
	possible types.  In the special case where $R$ is the successor
	relation $\succrel$ we replace it with each of $\succrel\spni$,
	$\succrel\stri$ and $\succrel\mix$.  Let $\tilde{P}$ be the set of
	rules that is obtained from $P$ by replacing each rule $\rho$ in $P$
	with all of its mixed versions.  A rule is called \e{type
		inconsistent} if it is inconsistent with respect to the type
	restrictions imposed by $\sigordplus$.  We omit from $\tilde{P}$
	rules $\rho$ that are \e{type inconsistent} and obtain $P^\prime$.
	Note that since we have omitted the type inconsistent rules
	$P^\prime$ is a $\sp$ program.  Since untyped encodings when viewed
	as mixed instances are consistent we obtain the desired result.
\end{proof}

Note that Lemma~\ref{lem:Datalog_to_SP} compares between a mixed
instance $P^\prime_Q(D^{\prime\prime})$ and a standard one
$Q^+(D^{\prime})$. However, this is well defined since the comparison
is done at the instance level.
\begin{example}
	This example is aimed to demonstrate the construction in the
	previous proof (of Lemma~\ref{lem:Datalog_to_SP}).  Let us consider
	the $\semiposdatalog$ program $P$ that contains the rule $R(x,y)\dla
	S(x), T(y,z)$.  The relation atom $R(x,y)$ has four different mixed
	versions, such as the following.
	\begin{itemize}
		\item $R_{\mathsf{str},\mathsf{str}}(\ol{x},\ol{y})$ wherein both
		attributes are string attributes.
		\item $R_{\mathsf{spn},\mathsf{str}}(x,\ol{y})$ wherein the first
		attribute is a span attribute and the second is a string attribute.
	\end{itemize}
	The rule $R(x,y)\dla S(x), T(y,z)$ has $2^5$ different mixed versions,
	one for each ``type assignment'' for its variables, such as the
	following.
	\begin{itemize}
		\item $R_{\mathsf{str},\mathsf{str}}(\ol{x},\ol{y})\dla
		S_{\mathsf{str}}(\ol{x}),
		T_{\mathsf{str},\mathsf{str}}(\ol{y},\ol{z})$
		\item $R_{\mathsf{spn},\mathsf{str}}(x,\ol{y})\dla
		S_{\mathsf{str}}(\ol{x}),
		T_{\mathsf{str},\mathsf{str}}(\ol{y},\ol{z})$
	\end{itemize}
	
	Note that these two are rules in the resulting $\sp$ program
	$\tilde{P}$.  However, the second rule is type inconsistent due to the
	variable $x$ that is regarded as a span variable in the head atom and
	as a string variable in the atom $S_{\mathsf{str}}(\ol{x})$, and thus
	is not a rule in $P^\prime$.
\end{example}

Note that the input of the $\sp$ program from
Lemma~\ref{lem:Datalog_to_SP} is a mixed encoding. We next show that
$\sp$ with stratified negation is expressive enough to construct the
mixed encoding of an instance over a span-free signature.
\begin{lemma}\label{lem:constructwithSP}
	There exists an $\sp$ program $P=\angs{\E,\I,\Phi,\Rout}$ such that
	$\E\cup\I$ contains $\sigordplus$ and the following holds. For all
	instances $D$ over $\E$ and relation symbols $R\in \sigordplus$ we
	have that $ R^{(\spencode{D})} = R^ {(P(D))}$.
\end{lemma}
\begin{proof}
	We construct the program $P$ as follows.
	
	For every relation symbol $R$ of $\E$ and
	$i=1,\dots,\arity{R}$ we use the following rules:
	\begin{align*}
	\Str(\ol{x}^i_y) \dla &
	R(\ol{x}^1, \ldots,\ol{x}^i,\ldots , \ol{x}^{\arity{R}}), 
	\angs{\ol{x}^i}\angs{ y\{.^* \}} \\
	\Spn(y) \dla &
	R(\ol{x}^1, \ldots,\ol{x}^i,\ldots , \ol{x}^{\arity{R}}), 
	\angs{\ol{x}^i}\angs{ y\{.^* \}} 
	\end{align*}

	For all $\sigma \in \Sigma$ we use the following rule:
	\[
	R_{\sigma}(\ol{x},y)\leftarrow \Str(\ol{x}), \angs{\ol{x}}\angs{ y\{ \sigma \} }
	\]

	In order to define the successor relation $\succrel\stri$, we
	define a strict total order $\succr{\stri}$, which is the
	usual lexicographic order. We denote our alphabet $\Sigma$ by
	$\{\sigma_1,\ldots,\sigma_n\}$.  In the usual lexicographic
	order, a string $\strs$ follows $\strs^\prime$ in this order
	if either $(1)$ $\strs^\prime$ is a strict prefix of $\strs$
	or $(2)$ the first symbol in which they differ is $\sigma_i$
	in $\strs$ and $\sigma_j$ in $\strs^\prime$ where $j<i$.  This
	can be expressed with the following $\sp$ rules using the
	binary relation $\rel{StrEq}$ that holds pairs of equivalent
	strings and can be expressed in $\sp$ (see the comment in
	Example~\ref{ex:splog} in the body of the paper).  For case
	$(1)$ we have:
	$$
	\succr{\stri}(\ol{x},\ol{x}^\prime)
	\dla \Str(\ol{x}), \Str(\ol{x}^\prime),
	\angs{\ol{x}} [ y\{.^* \}.^+ ] , 
	\angs{\ol{x}^\prime} [ y\{.^* \} ],
	\rel{StrEq}(\ol{x}_y, \ol{x}^\prime_y)
	$$
	And for case $(2)$:
	\begin{align*}
	\succr{\stri}(\ol{x},\ol{x}^\prime) \dla& \Str(\ol{x}), \Str(\ol{x}^\prime),
	\angs{\ol{x}} [ y\{.^* \} \cdot \sigma_i \cdot .^* ] , 
	\angs{\ol{x}^\prime} [ y\{.^* \} \cdot \sigma_j \cdot .^* ],
	\StrEq(\ol{x}_y,{{\ol{x}}^\prime}_y) 
	\end{align*}
	This rule is repeated for every $1 \le j<i\le n$.
	Based on $\succr{\stri}$,
	we use stratified negation to define the successor relation $\succrel{\stri}$.
	To do that, we define the binary relation $\rel{NotSucc}\stri$ that holds tuples $(\ol{x}^\prime, \ol{x})$ where $\ol{x}^\prime$ is not the successor of $\ol{x}$ with respect to $\succr\stri$.
	\begin{align*}
	\rel{NotSucc}{\stri}(\ol{x}_1,\ol{x}_2) \dla 
	\succr\stri(\ol{x}_1, \ol{x}_3),
	\succr\stri(\ol{x}_3, \ol{x}_2)
	\end{align*}
	and then,
	\begin{align*}
	\succrel\stri(\ol{x}, \ol{x}^\prime) \dla 
	\succr\stri(\ol{x}, \ol{x}^\prime), \neg \rel{NotSucc}{\stri}(\ol{x}, \ol{x}^\prime)
	\end{align*}
	Note that the first string in the lexicographic order is alway $\epsilon$ (since for every $D$, its extended active domain contains $\epsilon$). Therefore we have:
	\begin{align*}
	\first(\ol{x}) \dla \Str(\ol{x}), \angs{\ol{x}}[x\{ \epsilon \}]
	\end{align*}
	To define the relaion $\succrel\mix$ we need to find  the last string in the extended active domain of the input instance. 
	For this purpose, we define the relation $\last\stri$ as follows:
	\begin{align*}
	\StrNotLast (\ol{x}) \dla & \Str(\ol{x}), \Str(\ol{x}^\prime),  \succr{\stri}(\ol{x}^\prime, \ol{x})\\
	\last\stri(\ol{x})\dla& \Str(\ol{x}), \neg \StrNotLast(\ol{x})
	\end{align*}
	We can now define the relation $\succrel\mix$. Note that in the sequel we define a strict total order $\succ\spni$ on the spans in the extended active domain which is the lexicographic order (see
	Comment~\ref{com:ordsp} 
	in the body of the paper). The first span according to the lexicographic order is always $\mspan{1}{1}$ regardless of the input instance. We therefore use the following rule according to which the successor of the last string in the extended active domain of the input is the first span.  
	\begin{align*}
	\succrel\mix(y,\ol{x}) \dla \last\stri(\ol{x}),\angs{\ol{x}}[y\{\epsilon \} .^*]
	\end{align*}
	Similarly to the definition of $\succrel\stri$, we define $\succrel\spni$ by defining a strict total order $\succ\spni$ on the spans in the extended active domain.
	Note that the span $y$ follows $y^\prime$ in this order if either $(1)$ $y$ begins after $y^\prime$ begins or $(2)$ they both start in the same position but $y$ ends strictly after $y^\prime$ ends.  
	For case $(1)$ we have: 
	\begin{align*}
	\succ\spni(y,y^\prime) \dla \Str(\ol{x}), 
	\angs{\ol{x}}[ z\{.^* \} y\{ .^* \} .^* ], 
	\angs{\ol{x}}[ z^\prime\{.^* \} y^\prime \{ .^* \} .^* ],
	\angs{\ol{x}}[ z\{ z^\prime \{.^*\} .^+ \} .^* ],
	\end{align*}
	and for $(2)$:
	\begin{align*}
	\succ\spni(y,y^\prime) \dla \Str(\ol{x}), 
	\angs{\ol{x}}[ .^* y\{ y^\prime \{.^* \}  .^+ \} .^* ]
	\end{align*}
	The relation $\succrel\spni$ is defined based on $\succ\spni$
	in a similar way we have defined $\succrel\stri$ based on
	$\succ\stri$. Moreover, the relation $\last$ is defined based
	on $\succ\spni$ in a similar way $\last\stri$ was defined
	based on $\succ\spni$.  Therefore we skip these definitions.
\end{proof}

We can now conclude the proof of Theorem~\ref{thm:ptime-sp-decision}.
Assume that $Q$ is computable in polynomial time.  Let $D^{\prime}$ be
the untyped encoding of $D$ where $\iota$ from
Lemma~\ref{lem:spanExtendedEncodingRestoreI} is the identity.  Let
$D^{\prime\prime}$ be the mixed encoding that corresponds with
$D^{\prime}$.  Let $P^\prime_Q$ be the $\sp$ program obtained from
Lemma~\ref{lem:Datalog_to_SP}.  It holds that
$P^\prime_Q(D^{\prime\prime}) = Q^+(D^\prime)$.  Due to the definition
of $Q^+$ and since $\iota$ is the identity we have that
$P^\prime_Q(D^{\prime\prime}) = Q(D)$.  Due to
Lemma~\ref{lem:constructwithSP} we can construct $D^{\prime\prime}$
from $D$ using a $\sp$ program $P^\prime$.  Combining $P^\prime_Q$
with $P^\prime$ into a single $\sp$ program completes the proof.

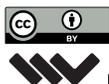
\end{document}